\documentclass{llncs}
\usepackage[utf8]{inputenc}
\usepackage{amsmath}
\usepackage{xcolor}
\usepackage{amssymb}
\usepackage{mathrsfs} 
\usepackage{enumitem}
\usepackage{hyperref}
\usepackage{cleveref}
\usepackage{algorithm}
\usepackage{algpseudocode}

\DeclareMathOperator{\E}{E}

\begin{document}

\title{Competition among Parallel Contests}

\author{
Xiaotie Deng  \inst{1} \and
Ningyuan Li \inst{1} \and
Weian Li \inst{1} \and
Qi Qi \inst{2}
}

\institute{Center on Frontiers of Computing Studies, School of Computer Science, Peking University, Beijing, China\\
\email{\{xiaotie, liningyuan, weian\_li\}@pku.edu.cn} \and
Gaoling School of Artificial Intelligence, Renmin University of China, Beijing, China\\
\email{qi.qi@ruc.edu.cn}
}

\maketitle

\begin{abstract}
We investigate the model of multiple contests held in parallel, where each contestant selects one contest to join and each contest designer decides the prize structure to compete for the participation of contestants.
We first analyze the strategic behaviors of contestants and completely characterize the symmetric Bayesian Nash equilibrium. 
As for the strategies of contest designers, when other designers' strategies are known, we show that computing the best response is NP-hard and propose a fully polynomial time approximation scheme (FPTAS) to output the $\epsilon$-approximate best response. When other designers' strategies are unknown, we provide a worst case analysis on one designer's strategy. We give an upper bound on the utility of any strategy and propose a method to construct a strategy whose utility can guarantee a constant ratio of this upper bound in the worst case.
\end{abstract}


\setcounter{footnote}{0}
\section{Introduction}
\label{section:intro}

Contest plays an irreplaceable role in today's life. For example, the innovation competitions on the online platforms, the reward mechanism within the company, and even the auctions can be viewed as contests in a sense. Because of this, the contest depicts a scene in which many players compete for several designed prizes, capturing many realistic game-theoretical settings involving competition, and is an important part of mechanism design theory, which has attracted the attention of many researchers from the past to the present. So far, most of the research literature in contest theory has focused on the setting of a single contest and aimed to design the rewarding policy to achieve some specific goals. For example, the initial paper \cite{MS01} aims to maximize the total output (or total effort) of the contestants. After that, other variants of single contests are also well studied.

With the emergence of crowdsourcing competitions, contests are becoming increasingly popular. More and more contests are run in parallel nowadays. For example, Amazon’s Mechanical Turk, an online crowdsourcing platform, has released thousands of outsourced tasks for individuals (in fact, each is a contest). Similarly, TopCoder holds about 6000 software contests every year, and practitioners compete for about \$10000 in awards. In addition, take another situation that our researchers are familiar with as an example. When we finish a paper, we will choose one conference from several conferences to submit.  Each conference can be considered as a contest and decides which papers can be accepted. From the author's perspective, they need to balance the chance of acceptance and the reputation of the target conferences. From the perspective of conference organizers, they think about how to set the acceptance criteria to attract high-quality papers. 

Compared with the mature results of single contest, the research on multiple contests held at the same time has received less attention. A recent survey on contest theory \cite{S20} mentioned that examining the economy of competitions among contests faces a series of challenges, especially in exploring the behavior of contest designers. In a single contest, we only need to consider the strategic behaviors of the contestants. While in the parallel contests setting, how the contestants choose the most suitable contest, and how the contest designer competes with other designers, including setting the reward structure and ranking policy should be both taken into consideration. 

In this paper, we consider the competition among multiple contests. Each contest designer sets its own prize structure within its budget. As in many practical applications, in order to prevent participants from being distracted, or unfairness, the number of contests they can participate in is limited. For instance, the Gates Foundation allows submission only to a single contest, the conferences allow submission only to a single conference, etc. Indeed, most of the literature on multiple contests assumes that each contestant can only participant in one contest (e.g., Azmat and M\"oller \cite{AM09}). In our model,  we also allow each contestant choose to participate in one contest only. Each designer's goal is to maximize the total value of the participants. Every contestant wants to maximize the prize she wins. There are two games in our model: One is among contestants who decide which contest to participate in; The other is among contest designers, who design prize structures to attract contestants.

\subsection{Main Contributions}
Our model offers some new features that contribute to the contest theory. First, while most previous studies have focused only on a single contest, we consider multiple contests. Second, in our model designers have different budgets and can design any reward policies. Third, the strategic behaviors of the contestants are also take into consideration.  These novelties increase the complexity of our analysis and require special technical attention. 

For the game of contestants, we concentrate on the contestants' behavior choosing the contest to join, and precisely characterize the symmetric Bayesian Nash equilibrium (BNE). 
For the game of contest designers, we analyze the strategies of contest designers (i.e., how to design the prize structure within the given budget) in two cases: the other designers' strategies are given or not given. When the other designers' strategies are given, we show that it is NP-hard to compute a best response by a reduction from a weighted knapsack problem, even in the case of two contests. Thus, we propose a fully polynomial time approximation scheme (FPTAS) to output an $\epsilon$-approximate best response. When the other designers' strategies are not known, we introduce the concept of ``safety level'' \cite{T02} to evaluate a prize policy, which is defined as the minimum utility this prize policy gets. 
An upper bound on the safety level of any policy is derived and taken as a benchmark, and we propose a construction of prize policy whose safety level can guarantee a constant ratio of this benchmark. 

Our technical contributions follow a framework of relaxation and rounding,  and can be summarized as follows:
\begin{itemize}
    \item \textbf{Relaxation.} We introduce two concepts: \emph{interim allocation function} and \emph{cumulative behavior}. Interim allocation function of one contest represents the expected prize one contestant with a certain skill can obtain in this contest. Cumulative behavior 
    is the probability that one contestant with a skill higher than a certain level joins one particular contest. We show that under a symmetric equilibrium of contestants, the utility of all contestants and all contest designers can be calculated from these two concepts. And we further characterize a symmetric equilibrium of contestants by deriving the condition for the cumulative behavior, which is essentially different from the traditional methodology which focuses on the first-order conditions and analyzes the mixed strategy directly.
    
    Our following techniques are mainly built on a relaxation step with the help of these two concepts. We relax the strategy space of a designer from the rank-by-skill allocation of $n$ prizes to an arbitrary interim allocation function on $[0,1]$. We call the former the original model and the latter the generalized model. We will show that this relaxation greatly eases the analysis of the game between the designers, for that we are allowed to freely adjust or construct the strategy of a designer as a function. \\
    \item \textbf{Technical simplifications and tools.} Under the generalized model, we derive some useful technical results. First, we can merge several contests into one by properly combining their interim allocation functions, without affecting the equilibrium behavior of contestants. Therefore, from a certain designer's view, all other designers can be merged as one bundle, so we can focus on the situation of $m=2$ without loss of generality. We remark that this also implies that there is no need for a contest designer to consider partitioning her budget to hold multiple contests.
    Second, we come up with a useful tool, \emph{Horizontal Stretching}, which adjusts the budget usage of an interim allocation function by stretching its graph horizontally.
    This allows a designer to imitate the strategy of her opponent, or to fit an existing strategy into her budget, guaranteeing a good utility. Those results play an important role in our analysis of the strategy of a designer.\\
    \item \textbf{Rounding.} Finally we propose a ``rounding" step to convert a strategy in the generalized model back into the original model. We show that an interim allocation function can be approximated by a rank-by-skill prize structure which can be implemented in the original model, and the utility loss tends to 0 when $n$ is sufficiently large. Therefore, a contest designer can freely design her strategy in the generalized model without hesitation.
\end{itemize}

\subsection{Related Works}
Our paper contributes to the literature in the scope of economics and computer science, particularly in topics of multiple contests competition and single contest design. Ghosh and Kleinberg \cite{GK16} consider a model, in which the strategy of contestant is whether or not to participate in a contest, rather than the effort level of contestant. Similarly, our paper investigates a general setting of multiple contests where the decision of contestants is to select which contest to join. Besides, the concept of ``safety level'' used to analyze the strategy of a designer in our paper, is proposed by Tennenholtz \cite{T02}. Recently, Lavi and Shiran-Shvarzbard \cite{LS21} introduce ``safety level'' into the field of contest theory and analyze the safety level of a particular kind of contest policy. However, in their model, the proposed policy restricts contestants' behaviors. In our paper, we consider a totally different model and characterize the contestants' behavior equilibrium under any policy. 

\subsubsection{Multiple Parallel Contests}
The current research in this field mainly focus on the equilibrium of contestants, when the mechanisms of prize allocation in the contests are fixed. Initially, Azmat and M\"{o}ller \cite{AM09} consider the two identical contests setting and uses Tullock model\cite{T01} as the ranking policy. They investigate the choice equilibrium of players and the prize structure equilibrium in different goals. DiPalantino and Vojnovic \cite{DV09} study the case of $m$ auction-based crowdsourcing contests. They give the equilibrium results of players in symmetric setting and extend the results to several asymmetric settings. B{\"u}y{\"u}kboyac{\i} \cite{B16} compares the environment of two contests with a single contest and show that the former one performs well when the abilities of contestants differ greatly. Azmat and M\"{o}ller \cite{AM18} investigate environment of two types of contest, high and low, but the contests use the similar prize structure. At the same period, the contestant's behavior is characterized in Morgan et al. \cite{MSV18}, when contestant faces two contests with different prize structures: one has high return accompanied by high competition; the other has low return and low competition. Hafalir et al. \cite{HHKK18} model the enrollment of university as a contest and demonstrates that students with lower ability prefer to join multiple contests, but students with high ability would prefer to join only one. Recently, Juang et al. \cite{JSY20} consider the setting of two contests with different prize amount. They show that the contests with higher prize can attract more contestants, but the effort of one contestant in a contest is irrelevant to the number of participants in that contest.  Deng et al. \cite{DGLLL21} study competition among contests with contestants having a public symmetric skill and show that the optimal contest in the monopolistic setting is also the equilibrium strategy in the competitive setting. In contrast, our contestants have private skills, which makes the problem significantly different. Their results no longer hold here. Our results are much different from theirs. K{\"o}rpeo{\u{g}}lu et al. \cite{KK22} focus on that one contestant can join several contests, but the output in each contest is affected by an uncertainty variable. Under this model, they prove that the more contests one contestant participates in, the better the utility of contest organizer becomes. Birmpas et al. \cite{BKL22} study the multiple crowdsourcing reviews, and they design a reward scheme to keep existence of equilibrium and consider the efficiency of equilibria. Elkind et al. \cite{EGG22} perform the equilibrium analysis of multiple equal sharing contests.  
The settings and problems studied in these works are generally different from ours, leading to fundamental differences in techniques and results.

\subsubsection{Single Contest} In the past decades, the research on single contest has become mature, which provides us with an abundance of techniques on analyzing the strategies of contest designers while facing the competition. A comprehensive survey can be found in \cite{C07}. Most papers in this field concentrate on the prize design under rank-by-output policy and the output is mainly decided by the ability and effort of contestants. In the initial period, Glazer and Hassin \cite{GH88} show the optimal contests for identical case and non-identical case, with respect to the ability of contestants, when the designer aims to maximize the total outputs. Moldovanu and Sela \cite{MS01} extend the above result and consider the case that the abilities of contestants are drawn from a publicly known distribution. They design the optimal contests for linear, concave, convex cost functions, respectively. Che and Gale \cite{KG03} investigate a procurement contests in symmetric and asymmetric cases and give the corresponding optimal contests. Archak and Sundararajan \cite{AS09} consider a crowdsourcing contest with a large number of participants and characterize the asymptotically optimal prize structure. Ghosh an Hummel \cite{GH12} aim to maximize a class of functions of the number of submissions and the qualities of outputs. 
Chawla et al. \cite{CHS19} propose the optimal crowdsourcing contest with the methodology of auction theory. Elkind et al. \cite{EGG21} study the threshold objective of designer, and come up with the optimal contests. A series of papers \cite{MS06,TX08,KC18} take the productivity of contestants into account and run the rank-by-output policy. 
Furthermore, many papers \cite{KG03,LSR17,LS18,LSA19,LSH20} discuss about other variants of single contest. 
In our paper, we do not consider the exertion of effort in existing papers and assume that designers can judge the skill of contestants and maximize the values from attracted contestants. 

As for our paper, we consider not only the equilibrium behaviors of contestants, under any prize structures of several contests, but also the design of strategies for contest designers under the competition among multiple contests.

\subsection{RoadMap} 
In Section \ref{section:pre}, we introduce our model and provide the definition of cumulative behavior and interim allocation function, and other necessary notions. In Section \ref{section:players_equilibrium}, the equilibrium behavior of contestants is characterized completely. In Section \ref{sec:relax and simp}, we relax the feasible strategy space of contest designers and propose some technical results. Based on the results in Section \ref{section:players_equilibrium} and \ref{sec:relax and simp}, in Section \ref{section:best_response}, we investigate the best response of one contest designer and design an FPTAS to output an approximate best response. In Section \ref{section:safetylevel}, we focus on the safety level of a first-mover strategy, and propose the construction of a strategy with a constant-competitive safety level. Furthermore, in Section \ref{section:implement}, we show a rounding technique to find a discrete prize structure in original feasible region with a little loss on designer's utility. In Section \ref{section: conclusion}, we give a summary of the whole paper and propose the directions for future work. Due to space limitations, we put most of the proofs in appendix.

\section{Model and Preliminary}
\label{section:pre}

In this section, we introduce our model formally. Assume that there are $m$ contests and $n$ contestants. To avoid ambiguity, we let $j \in [m] = \{1,2,\cdots, m\}$ and $i \in [n] = \{1,2,\cdots, n\}$ denote a contest and a contestant, respectively. 

Each contest designer $j\in[m]$ has a publicly known (average) budget constraint $t_j$, constraining that the total amount of the prizes in contest $j$ cannot exceed $nt_j$. The prize structure of contest $j$ is $\vec{w}_j = (w_{j1}, w_{j2}, \cdots, w_{jn})$, representing the prize amounts of $n$ prizes, where $w_{j1}\geq w_{j2}\cdots\geq w_{jn}\geq 0$ and $\sum_{i=1}^n w_{ji}\leq nt_j$. This notation of average budget $t_j$ is for convenience of future discussion, and we will refer to $t_j$ as designer $j$'s budget.

We call this the original model, where a designer's strategy is the amount of the $n$ prizes in the rank-by-skill allocation, to distinguish from the generalized model introduced later, where the strategy space of a contest designer is relaxed with the help of interim allocation function.

Each contestant $i\in[n]$ has a private skill $s_i$, which is identically and independently drawn from a publicly known distribution $G$. For simplicity we assume $G$ is a continuous distribution, then the cumulative distribution function (CDF) $G(s)$ is continuous. For convenience, we use the notation of \emph{quantile} to represent the skill. Specifically, given $s \sim G$, we define the quantile $q=1-G(s)$. If we exploit the quantile $q$ as the random variable, we can map it to the skill by the function $s(q)=G^{-1}(1-q)$. Note that $s(q)$ is a strictly decreasing function of $q$, and whatever the distribution $G$ is, $q$ always follows the $U[0,1]$ distribution. In the following parts, we use $q_i$ to represent the contestant $i$'s competitiveness, directly.


Our model consists of three stages:
\begin{enumerate}
    \item[$\bullet$]In the first stage (Game of Designers, or GoD), all contest designers announce their prize structures $\{\vec{w}_{1}, \vec{w}_{2}$ $, \cdots, \vec{w}_{m}\}$ to the contestants.
    \item[$\bullet$] In the second stage (Game of Contestants, or GoC), given all prize structures of the contests, each contestant $i\in[n]$ decides which contest to participate in, denoted by $J_i \in [m]$. 
    
    We allow the contestants to take mixed strategies. The mixed strategy of contestant $i$ is defined as $\vec{\pi}_{i}(q_i)=(\pi_{i,1}(q_i),\pi_{i,2}(q_i),\cdots,\pi_{i,m}(q_i))$, where $\pi_{i,j}(q_i)$ is the probability that contestant $i$ participates in contest $j$ (i.e., takes the pure strategy $J_i=j$) and $\sum_{j=1}^{m}\pi_{i,j}(q_i)=1$. Let $\mathcal{S}$ be the mixed strategy space containing all such $\vec{\pi}(q)$.
    \item[$\bullet$] After the decision stages, each contest designer $j$ gets informed of the list of participants in her contest, denoted by $I_j=\{i:J_i=j\}$. In the third stage, each contest designer $j$ executes a rank-by-skill allocation, that is, she ranks the contestants in $I_j$ by their quantiles in ascending order and awards her prizes to the corresponding contestants\footnote{Note that if $|I_j|$ is less than $n$, only the first $|I_j|$ prizes are handed out.}: the prize $w_{j1}$ is awarded to the contestant with the lowest quantile (representing the highest skill), the prize $w_{j2}$ is awarded to the contestant with the second lowest quantile, and so on. 
\end{enumerate}

Keeping the above process in mind, we first define the utility of a contestant. Let $\boldsymbol{J}=(J_1,J_2,\cdots,J_n)$ and $\boldsymbol{\pi}=\left(\vec{\pi}_{1}(q_1),\vec{\pi}_2(q_2),\cdots,\vec{\pi}_{n}(q_n)\right)$ represent the pure strategy profile and the mixed strategy profile. Under the current pure strategy profile $\boldsymbol{J}$, the utility of contestant $i$ is exactly the prize she received in the contest $J_i$, that is,
$$u_i(\boldsymbol{J})=w_{J_i,\mathrm{rank}(i,J_i)},$$
where $\mathrm{rank}(i,J_i)=|\{i'\neq i: (q_{i'}\leq q_i)\land (J_{i'}=J_i)\}|+1$.

Since the game of contestants is of incomplete information, where the quantile of each contestant is private information, we mainly focus on the symmetric Bayesian Nash equilibrium in the GoC. Assume that all contestants follow a same mixed strategy $\vec{\pi}(q) = (\pi_1(q),\pi_2(q),\cdots,\pi_m(q)) \in \mathcal{S}$. In other words, for any $i \in [n]$, we assume $\vec{\pi}_i(q_i)= \vec{\pi}(q_i)$. Before giving the definition of symmetric BNE, we first introduce the concept of \emph{cumulative behavior} which helps us formulate the expected utility gained from a contest. 
\begin{definition}[Cumulative Behavior]
\label{definition:cumulative-behavior}
Assume that all contestants adopt the strategy $\vec{\pi}(q)=(\pi_1(q),\pi_2(q),\cdots,\pi_m(q))\in\mathcal S$ symmetrically. For each $j\in[m]$, define $H_{\pi_j}(q):[0,1]\to[0,1]$ as $$H_{\pi_j}(q):=\int_0^q\pi_j(u)du.$$
We call $\vec{H}(q)=(H_{\pi_1}(q),H_{\pi_2}(q),\cdots,H_{\pi_m}(q))$ a cumulative behavior. We also say that $\vec{H}(q)$ is the cumulation of $\vec{\pi}(q)$. 
\end{definition}
Note that, when a strategy $\vec{\pi}$ is adopted by all contestants symmetrically, $H_{\pi_j}(q)$ is the probability of the event that, for a contestant $i$, her quantile $q_{i}$ turns out to be less than $q$, and her pure strategy turns out to be $J_i=j$. 
Formally, given any $j\in[m]$ and $q\in[0,1]$, for any $i\in[n]$,
\begin{align}
    \Pr[(q_i\leq q)\land (J_i=j)]=H_{\pi_j}(q).\label{eq:probability-meaning-of-H}
\end{align}
Observe that the prize that contestant $i$ receives from contest $j$ completely depends on the number of contestants ranking before $i$, or formally $\mathrm{rank}(i,j)-1$, which follows the binomial distribution $B(n-1,H_{\pi_j}(q_i))$. Therefore, the expected utility received by contestant $i$ in contest $j$ can be written as
\begin{align}
\E_{q_{-i}}[u_i|q_i=q,J_i=j]=\sum_{k=1}^n w_{j,k}\binom{n-1}{k-1}H_{\pi_j}(q)^{k-1}(1-H_{\pi_j}(q))^{n-k}.\label{eq:expected-utility-in-rank-order}
\end{align}
In addition, from the perspective of the designer, \Cref{eq:expected-utility-in-rank-order} shows that in any contest $j$, the expected prize allocated to a contestant with quantile $q$ is determined by $H_{\pi_j}(q)$. Inspired by this observation, we introduce the concept of interim allocation function $x_j(h)$ to denote the expected prize  of contest $j$, which will be used throughout the whole paper.
\begin{definition}[Interim Allocation Function]
\label{definition:interim-allocation-function}
For any non-increasing and non-negative function $x_j(h)$ on $[0,1]$, we say $x_j(h)$ is an interim allocation function representing the prize policy of contest $j$, if under any symmetric strategy $\vec{\pi}(q)=(\pi_1(q),\cdots,\pi_m(q))$, for any $i\in[n]$ and any $q\in[0,1]$, it holds that $$\E_{q_{-i}}[u_i|q_i=q,J_i=j]=x_j(H_{\pi_j}(q)).$$
\end{definition}
Using this concept, the prize structure of each contest $j$ is now represented by an interim allocation allocation function 
$x_j(h)$\footnote{From \Cref{eq:expected-utility-in-rank-order}, it is not difficult to see that any contest $j$ with the rank-by-skill allocation has the interim allocation function $x_j(h)=\sum_{k=1}^n w_{j,k}\binom{n-1}{k-1}h^{k-1}(1-h)^{n-k}$, which is continuous, non-increasing and non-negative. In addition, since $\int_0^1x_j(h)dh= (\sum_{k=1}^n w_{j,k})/n$, the budget constraint on the interim allocation function can be written as $\int_0^1x_j(h)dh\leq t_j$.}.

Given these two concepts of cumulative behavior and interim allocation functions, we can define the symmetric Bayesian Nash equilibrium:
\begin{definition}[Symmetric Bayesian Nash Equilibrium]
\label{definition:contestants-symmetric-equilibrium}
We say $\vec{\pi}^*(q)=(\pi^*_1(q),\cdots,\pi^*_m(q))$ $\in\mathcal S$ is a symmetric Bayesian Nash equilibrium, if when $\vec{\pi}^*$ is adopted as the symmetric strategy, i.e. $\vec{\pi}_i=\vec{\pi}^*$ for all $i\in[n]$, it holds for any $q\in[0,1]$ that
\begin{align*}
    \{j\in[m]:\pi^*_j(q)>0\}\subseteq\arg\max_{j\in[m]}\E_{q_{-i}}[u_i|q_i=q,J_i=j].
\end{align*}
Given interim allocation functions $x_1(h),\cdots,x_m(h)$, this is equivalent to
\begin{align*}
    \{j\in[m]:\pi^*_j(q)>0\}\subseteq\arg\max_{j\in[m]}x_{j}(H_{\pi^*_{j}}(q)).
\end{align*}
If $\vec{\pi}^*(q)$ is a symmetric Bayesian Nash equilibrium, its cumulation $\vec{H}^*(q)$ is called a cumulative equilibrium behavior (CEB).
\end{definition}

Lastly, we propose the utility of contest designers. In our model, each contest designer may have a different evaluation criterion on the value of contestants. We assume that, for each contest designer $j$, the value of any contestant is a non-decreasing function of the contestant's skill or a non-increasing function of the contestant's quantile, denoted by $v_j(q)$. We also suppose that each $v_j(q)$ is non-negative and bounded on $[0,1]$. The utility of contest designer $j$, under the pure strategy profile $\boldsymbol{J}$, is the total value of contestants participating in contest $j$, i.e., 
\begin{align*}
    R_j(\boldsymbol{J})=\sum_{i\in I_j}v_j(q_i).
\end{align*}
And when contestants use the symmetric mixed strategy $\vec{\pi}$, the designer $j$'s expected utility is
\begin{align}
    R_j(\vec{\pi})&=\E_{q_1,\cdots,q_n}\left[\sum_{i=1}^n v_j(q_i)\pi_j(q_i)\right]
    =n\E_{q\sim U[0,1]}\left[v_j(q)\pi_j(q)\right]
    =n\int_0^1v_j(q)dH_{j}(q),\label{eq:designer-expected-utility-StieltjesIntegral}
\end{align}
where the last equation holds by writing $\int_0^1v_j(q)\pi_j^*(q)dq$ in the form of Riemann-Stieltjes integral.





\section{The Equilibrium Behavior of Contestants}
\label{section:players_equilibrium}

In this section, we characterize the symmetric BNE of the contestants. Unlike the traditional methodology in game theory, which focuses on calculating the first order condition of the expected utility, we propose an essentially new method drawing support from the cumulative behavior. 

To help understanding our method profoundly, we briefly discuss the equilibrium in the complete information setting. Suppose that the competitiveness $q_i$ of each contestant is publicly known, so every contestant knows her ranking among all $n$ contestants. Given all prize structures $w_{1}, w_2, \cdots, w_{m}$, the Nash equilibrium $(J_1,\cdots,J_n)$ is straightforward: sort all the $mn$ prizes in descending order, and sort all contestants in ascending order by their quantile, then, the NE is that the contestant with the lowest quantile join the contest with the highest prize, the contestant with the second lowest quantile join the contest with the second highest prize and so on. The key point of the above process is that the contestant knows her own ranking and the exact prize she can obtain in each contest, when the strategies of all contestants with lower quantiles are known. 
Roughly speaking, our methodology can be viewed as a continuous generalization of the above process in the incomplete information setting.

Coming back to our incomplete information model where the skills are drawn i.i.d from a publicly known distribution, we consider the expected prize for a contestant $i$ in some contest $j$. Observe that the behavior of a contestant with the higher quantile can never affect the ranking of a contestant with the lower quantile. Consequently, one contestant only cares about the behavior of contestants with lower quantiles than hers. Intuitively, one can imagine a continuous version of the sequential decision making process: when $q$ moves from $0$ to $1$, a contestant with quantile $q$ calculates the expected remaining prize in each contest $j$ with the help of the cumulative behavior, then chooses the contests with the highest expected remaining prize to participate in. 

By Definition \ref{definition:cumulative-behavior} and \ref{definition:interim-allocation-function}, we know that cumulative behavior $\vec{H}(q)$ is the cumulation of mixed strategy $\vec{\pi}(q)$, and also decides the expected prize from each contest. Thus in our characterization of the equilibrium behavior, the cumulative behavior is a pivotal concept which serves as a bridge between the mixed strategy $\vec{\pi}(q)$ and the equilibrium condition.

We first discuss the relationship between a strategy $\vec{\pi}(q)$ and its cumulation $\vec{H}(q)$, claiming that $\vec{H}(q)$ is a good proxy of $\vec{\pi}(q)$, and it will be more concise to give the condition for a cumulative behavior $\vec{H}(q)$ to be a cumulative equilibrium behavior, instead of describing a symmetric equilibrium strategy $\vec{\pi}(q)$ directly. First, we point out that different mixed strategies can have the same cumulation. Given a mixed strategy $\vec{\pi}(q)$ and its cumulation $\vec{H}(q)$, for another mixed strategy $\vec{\tau}(q)$, if $\vec{\pi}(q)$ and $\vec{\tau}(q)$ disagree only on a set of measure zero, i.e., $\{q\in[0,1]:\exists j\in[m],\pi_j(q)\neq\tau_j(q)\}$ is of measure zero, then we know the cumulation of $\vec{\tau}(q)$ is also $\vec{H}(q)$. In fact, this is a necessary and sufficient condition. Note that when $\vec{\pi}(q)$ and $\vec{\tau}(q)$ disagree only on a set of measure zero, and when $q$ is drawn from $U[0,1]$, it happens with probability $1$ that $\vec{\pi}(q)=\vec{\tau}(q)$. Thus, $\vec{\pi}$ and $\vec{\tau}$ can be naturally regarded as the same behavior when they have the same cumulation.

Next we show how to reconstruct a strategy $\vec{\pi}(q)$ when given a cumulative behavior $\vec{H}(q)$, so that $\vec{H}(q)$ is the cumulation of $\vec{\pi}(q)$. If $\vec{H}(q)=(H_1(q),H_2(q),\cdots,H_m(q))$ is the cumulation of some $\vec{\pi}(q)\in\mathcal S$, one can easily verify that $\vec{H}(q)$ satisfies that\begin{enumerate}
    \item For any $j\in[m]$, $H_j(q)$ is a continuous, non-negative and non-decreasing function on $[0,1]$.
    \item For any $j\in[m]$ and $q\in[0,1]$, $\sum_{j=1}^mH_j(q)=q$.
\end{enumerate}
Reversely, let $\mathcal H$ denote the class of all $\vec{H}(q)$'s satisfying the two conditions, then the following proposition show that, $\mathcal H$ exactly consists of all possible cumulative behaviors, and we can find a corresponding $\vec{\pi}(q)$ for any $\vec{H}(q)\in\mathcal H$.

\begin{proposition}\label{proposition:differentiating-cumulative-behavior}
For any $\vec{H}(q)=(H_1(q),H_2(q),\cdots,H_m(q))\in\mathcal H$, there is some $\vec{\pi}(q)\in\mathcal S$ such that for all $j\in[m]$ and any $q\in[0,1]$, $H_j(q)=H_{\pi_j}(q)$. Specifically, given any $x_1(h),\cdots,x_m(h)$, $\vec{\pi}(q)$ can be constructed as
\begin{align}\label{eqn:h to pi}
\pi_j(q)=\begin{cases}
H_j'(q),&\text{if }q\in[0,1]\setminus E;\\
\frac1{\left|\arg\max_{j'\in[m]}x_{j'}(H_{j'}(q))\right|},&\text{if }q\in E\land j\in\arg\max_{j'\in[m]}x_{j'}(H_{j'}(q));\\
0,&\text{if }q\in E\land j\notin\arg\max_{j'\in[m]}x_{j'}(H_{j'}(q)),
\end{cases}
\end{align}
where $E=\{0,1\}\cup\bigcup_{j=1}^mE_j$, and $E_j$ is the set of non-differentiable points of $H_j(q)$ in $(0,1)$, which has measure zero.
\end{proposition}

\Cref{proposition:differentiating-cumulative-behavior} shows that the mapping from a mixed strategy $\vec{\pi}(q)$ to its cumulation $\vec{H}(q)$ is a many-to-one surjection from $\mathcal S$ to $\mathcal H$. Since different $\vec{\pi}(q)$ may have the same cumulation, it will be more concise to give the condition for a cumulative behavior $\vec{H}(q)\in\mathcal H$ to be a cumulative equilibrium behavior, instead of describing a symmetric equilibrium strategy $\vec{\pi}(q)$ directly.

We define some important notations, before describing the condition for a cumulative equilibrium behavior.
Let $x_j^{-1}(x):=\sup\{h\in[0,1]:x_j(h)\geq x\}$\footnote{For $x>x_j(0)$, we define $x_j^{-1}(x)=0$. This applies to future notations similarly.}, $Q(x):=\sum_{j=1}^m x_j^{-1}(x)$, and $Q^{-1}(q):=\max\{x:Q(x)\geq q\}$. Note that any $x_j^{-1}(x)$, $Q(x)$, and $Q^{-1}(q)$ are non-increasing and left-continuous, since each $x_j(h)$ is non-increasing on $[0,1]$.

To help understanding, recall the complete information setting. $x_j^{-1}(x)$ is analogous to how many prizes of at least $x$ can be offered in the prize structure of contest $j$. $Q(x)$ can be understood as how many prizes of at least $x$ are prepared in total by all designers, which is also the maximum number of contestants who can receive a prize of at least $x$. And $Q^{-1}(q)$ is analogous to the prize ranked approximately $qn$ among all $mn$ prizes. 
Intuitively, when the contestant with quantile $q$ is making decision, let $X=Q^{-1}(q)$, then the highest amount of the expected remaining prize in all the contests is $X$. Therefore, she will get an expected prize of $X$. 

Now we give the characterization of CEB in the following theorem, under the condition that all the interim allocation functions are strictly decreasing, which holds when every designer adopt a rank-by-skill prize structure where $w_{j,1}>w_{j,n}$, i.e., the $n$ prizes are not all equal. This condition guarantees that every $x_j^{-1}(x)$ is a continuous function. In the case that $w_{j,1}=w_{j,2}=\cdots=w_{j,n}$ for some contest $j$, $x_j(h)$ becomes a constant, and $x_j^{-1}(x)$ becomes not continuous, resulting in the need for more discussion. The full version of this theorem, which is compatible with interim allocation functions that may have constant intervals and discontinuities, is stated in \Cref{theorem:player-equilibrium-general-condition}.

\begin{theorem}
\label{theorem:rankbyskill-player-equilibrium-general-condition} Suppose $x_1(h),\cdots,x_m(h)$ are all strictly decreasing on $[0,1]$. Define $\underline{X}=Q^{-1}(1)$ and $\overline{X}=\max_{j} x_j(0)$ as the lower bound and the upper bound on prize level. 
For any $\vec{H}^*(q)\in\mathcal H$, 
$\vec{H}^*(q)$ is a cumulative equilibrium behavior if and only if the following condition holds: For any $X\in [\underline{X},\overline{X}]$, let 
$q=Q(X)$, then it holds for each $j\in[m]$ that $H_j^*(q) = x_j^{-1}(X)$.
    
\end{theorem}

\vspace*{-10pt}
\section{Relaxation and Simplification}
\label{sec:relax and simp}


In Section \ref{section:players_equilibrium}, we have characterized the contestants' equilibrium behavior given all prize structures of contests, with the help of cumulative behavior and interim allocation function. In this section, we first introduce the relaxation on the possible allocation rule of the contests. As we have shown, any rank-by-skill allocation can be represented by an interim allocation function. Starting from this section (Subsection \ref{subsec:gene inter allo}), we generalize the prize allocation system of each contest $j$ to an abstract allocation represented by the interim allocation function $x_j(h)$, which is only required to be non-negative and non-increasing. Under this generalized model, we define generalized symmetric equilibrium of the contestants and give the full description in \Cref{theorem:player-equilibrium-general-condition}. This generalization brings great convenience to our analysis on the strategies of the designers. Specially, we develop some techniques in Subsection \ref{subsec:simp contest envm}-\ref{subsec: HS}, which are frequently applied in \Cref{section:best_response} and \Cref{section:safetylevel}. 

\subsection{Relaxation of Interim Allocation}
\label{subsec:gene inter allo}
We have introduced the concept of interim allocation function, which represents the allocation rule of a contest. Now we generalize the prize structure of each contest $j$ from the rank-by-skill allocation to an abstract\footnote{We use the word abstract to emphasize that an interim allocation function $x_j(h)$ in the generalized model may not be implementable by rank-by-skill allocation in the original model. We propose how to implement it approximately in \Cref{section:implement}.} allocation represented by the interim allocation function $x_j(h)$.

From now on, we consider the generalized model, where we only assume that for each $j\in[m]$, $x_j(h)$ is a non-increasing, non-negative function on $[0,1]$, representing the prize policy of contest $j$. The budget constraint can be naturally\footnote{From \Cref{eq:expected-utility-in-rank-order}, we have seen that, any contest $j$ with a rank-by-skill allocation has the interim allocation function $x_j(h)=\sum_{k=1}^n w_{j,k}\binom{n-1}{k-1}h^{k-1}(1-h)^{n-k}$, and $\int_0^1x_j(h)dh= (\sum_{k=1}^n w_{j,k})/n\leq t_j$.} defined for the interim allocation function as $\int_0^1x_j(h)dh\leq t_j$.

Define the notation $\mathcal F_{t}$ as the set of feasible strategies of a designer within budget $t$, which is the class of functions $x(h)$ on $[0,1]$ which is non-negative, non-increasing, and satisfies that $\int_0^1 x(h)dh\leq t$.

When some interim allocation functions have discontinuous points, the symmetric equilibrium defined in \Cref{definition:contestants-symmetric-equilibrium} may not exist. We adapt the definition of symmetric equilibrium to this situation, using the right limit at each discontinuous point. With this generalization, it can be seen that the symmetric equilibrium always exists, as implied by \Cref{theorem:player-equilibrium-general-condition}.



\begin{definition}[Generalized Symmetric Bayesian Nash Equilibrium (GSBNE)]
\label{definition:contestants-generalized-symmetric-equilibrium}
Given $x_1(h),\cdots,x_m(h)$, we say $\vec{\pi}^*(q)\in\mathcal S$ is a generalized symmetric Bayes-Nash equilibrium, if for any $j\in[m]$ and any $q\in[0,1]$, it holds that
\begin{equation*}
    \{j\in[m]:\pi^*_j(q)>0\}\subseteq\arg\max_{j'\in[m]}x_{j'}(H_{\pi^*_{j'}}(q)+0)\footnote{We define the notions $f(x+0) = \lim_{x'\to x+0}f(x')$ and $f(x-0) = \lim_{x'\to x-0}f(x')$ as the right limit and the left limit of $f(x)$ respectively.}.
\end{equation*}
W.l.o.g, we define $x_j(1+0)=x_j(1)$, so that $x_{j'}(H_{\pi^*_{j'}}(q)+0)$ is always well-defined.

If $\vec{\pi}^*(q)$ is a generalized symmetric Bayesian Nash equilibrium, its cumulation $\vec{H}^*(q)$ is called a generalized cumulative equilibrium behavior (GCEB).
\end{definition}
Note that the definition of GSBNE is compatible with the definition of symmetric BNE (\Cref{definition:contestants-symmetric-equilibrium}) when all interim allocation functions are continuous. 
Now we extend \Cref{theorem:rankbyskill-player-equilibrium-general-condition} to completely characterize the generalized cumulative equilibrium behavior.
\begin{theorem}
\label{theorem:player-equilibrium-general-condition}
Define $\underline{X}=Q^{-1}(1)$ and $\overline{X}=\max_{j} x_j(0)$.  
$\vec{H}^*(q)$ is a generalized cumulative equilibrium behavior if and only if the following condition holds: $\forall X\in [\underline{X},\overline{X}]$, define $q_1=Q(X+0)$ and 
$q_2=Q(X)$. For each $j\in[m]$ 
and any $q\in[q_1,\min(q_2,1)]$, $$x_j^{-1}(X+0)\leq H_j^*(q)\leq x_j^{-1}(X).$$
\end{theorem}

\subsection{Simplification of Contest Environment}
\label{subsec:simp contest envm}

From a certain contest designer $j$'s perspective, the number of contests $m$ can be reduced to two, by viewing other designers as one bundle, and combining their interim allocation functions. 
Specifically, given the budgets $t_{j'}$ and strategies $x_{j'}(h)$ of all other designers $j'\neq j$, we define $t_{-j}= \sum_{j' \neq j}t_{j'}$ and $x_{-j}(h)=\max\{x:\sum_{j'\neq j}x_{j'}^{-1}(x)\geq h\}$, then the rest $m-1$ designers are bundled as one designer with budget $t_{-j}$ and the interim allocation function is $x_{-j}(h)$.
Note that $x_{-j}(h)$ is feasible within budget $t_{-j}$, since $\int_0^1 x_{-j}(h)dh=\int_0^{\overline{X}}x_{-j}^{-1}(X)dX=\sum_{j'\neq j}\int_0^{\overline{X}}x_{j'}^{-1}(x)dx=\sum_{j'\neq j}\int_0^1x_{j'}(h)dh\leq\sum_{j'\neq j}t_{j'}= t_{-j}$. 

Now we show that this simplification does not affect the contestant's behavior from designer $j$'s view. Given $x_1(h),\cdots,x_m(h)$, let $\mathcal H^*(x_1(h),\cdots,x_m(h))$ denote the set of all generalized cumulative equilibrium behaviors, and let $\mathcal H^*_j(x_1(h),\cdots,x_m(h))$ denote the set of $H_j(q)$ that is the $j$-th component of some $\vec{H}\in\mathcal H^*(x_1(h),\cdots,x_m(h))$. We have the following proposition. 

\begin{proposition}
\label{proposition:m-contest-equivalent-with-two-contest}
Given $x_1(h),\cdots,x_m(h)$, for any GCEB, $\vec{H}(q)=(H_1(q), H_2(q),\cdots, H_m(q))\in\mathcal H^*(x_1(h),\cdots,x_m(h))$ and any $j\in[m]$, if we define $\tilde{H}(q)=(H_j(q),q-H_j(q))$, then $\tilde{H}(q)\in\mathcal H^*(x_j(h),x_{-j}(h))$. In other words, $H_j(q)$ appears in the GCEB $\tilde{H}(q)$ in the GoC game with only two contest designers, whose interim allocation functions are $x_j(h)$ and $x_{-j}(h)$, respectively.
\end{proposition}

With this proposition, we can simplify the environment from $m$ contests to two contest. In the remaining section, except for the special case, we directly consider the GoD between two designers, designer $j$ and designer $-j$.

\subsection{The Best and Worst Equilibrium for A Certain Designer}
\label{subsec:charac best and worst equ}
When $x_j(h)$ and $x_{-j}(h)$ are not restricted to be strictly decreasing, the GCEB in $\mathcal H^*(x_j(h),x_{-j}(h))$ can be non-unique. Consequently, designer $j$ may get different utilities under different equilibrium behaviors of contestants. To deal with this non-uniqueness, we define the worst utility and the best utility of designer $j$ respectively, as 
\begin{align*}
    &R^{Worst}(x_j(h)|x_{-j}(h),v_j(q))=\min_{H_j\in \mathcal H_j^*(x_j(h),x_{-j}(h))}\int_0^1 v_j(q)dH_j(q), ~ \text{and}\\
    &R^{Best}(x_j(h)|x_{-j}(h),v_j(q))=\max_{H_j\in \mathcal H_j^*(x_j(h),x_{-j}(h))}\int_0^1 v_j(q)dH_j(q).
\end{align*}
Then, we give the exact characterization of the best and the worst $H^*_j(q)\in\mathcal H^*_j(x_j(h),x_{-j}(h))$, which are from the cumulative equilibrium which maximizes/minimizes $j$'s utility, respectively. Roughly speaking, the non-uniqueness of equilibria is due to the coincidence of discontinuous points of $x_j^{-1}(X)$ and $x_{-j}^{-1}(X)$\footnote{We note that this happens even in the original model, when there is two contests $j_1,j_2$ such that $w_{j_1,1}=w_{j_1,2}=\cdots=w_{j_1,n}=w_{j_2,1}=w_{j_2,2}=\cdots=w_{j_2,n}$.}, which indicates that both $x_j(h)$ and $x_{-j}(h)$ are equal to a certain constant $X$ on some positive-lengthed interval. As a result, the quantile of the contestants who receive expected prize $X$ can vary in a range. Intuitively, the best utility of designer $j$ is achieved when contestants with the best possible skills join the contest $j$. On the other hand, the worst case is that the contestants with the lowest skills choose the contest $j$. We denote $H_j^{Worst}(q)$ and $H_j^{Best}(q)$ as the worst and best equilibrium behaviors for designer $j$, respectively, and characterize them in the following proposition.
\begin{proposition}
\label{proposition:best-response-worst-and-best-equilibrium}
Given $x_j(h)$ and $x_{-j}(h)$, $H_j^{Worst}(q)$ and $H_j^{Best}(q)$ can be calculated as following:
For every $q\in[0,1]$, let $X=Q^{-1}(q)$, 
\begin{align*}
H_j^{Worst}(q)&=\max\{x_j^{-1}(X+0), q-x_{-j}^{-1}(X)\}
\end{align*} 
and
\begin{align*}
H_j^{Best}(q)&=\min\{x_j^{-1}(X),q-x_{-j}^{-1}(X+0)\}.
\end{align*}
For any $v_j(q)$, we have $R^{Worst}(x_j(h))=\int_0^1v_j(q)dH_j^{Worst}(q)$ and $R^{Best}(x_j(h))= \int_0^1v_j(q)dH_j^{Best}(q)$.
\end{proposition}
\subsection{Horizontal Stretching}
\label{subsec: HS}
Technically, we come up with a useful tool called \emph{Horizontal Stretching} to help the construction of advantageous strategy. Horizontal stretching allows a designer to imitate the strategy of her opponent, or to take advantage of an existing strategy which exceeds her budget, to get a guaranteed utility. All she needs to do is to horizontally scale the interim allocation function to fit it into her budget constraint.
The intuition is that, to reduce the expense, rather than lowering the prize amount to cover more contestants, which may cause a complete loss of the high-skilled contestants, it seems to be more efficient to reduce the number of prize winners, while slightly raising the prize level to compete for contestants with higher value. 

To formalize the idea of horizontal stretching, we introduce the concept of \textit{$C$-dominating} and \textit{$C$-strongly-dominating}.
\begin{definition}
Let $x(h),\hat{x}(h)$ be two interim allocation functions. Given $C>0$, we call $x(h)$ $C$-dominates $\hat{x}(h)$, if
$$x(Ch)\geq \hat{x}(h),\forall h\in[0,\min\{1,\frac1C\}].$$
Moreover, $x(h)$ $C$-strongly-dominates $\hat{x}(h)$, if
$$x(Ch)> \hat{x}(h),\forall h\in[0,\min\{1,\frac1C\}].$$
\end{definition}

It can be seen that, given some $\hat{x}(h)$, if we horizontally stretch $\hat{x}(h)$ by a factor of $C$, to construct $x(h)$ such that $x(Ch)=\hat{x}(h)$, then $x(h)$ $C$-dominates $\hat{x}(h)$, while the budget requirement is also scaled with factor $C$, i.e. $\int_0^1x(h)dh=C\int_0^1\hat{x}(h)dh$. Moreover, for any $\epsilon>0$, we have $x(h)+\epsilon$ $C$-strongly-dominates $\hat{x}(h)$, while $\int_0^1(x(h)+\epsilon)dh=C\int_0^1\hat{x}(h)dh+\epsilon$.

If we horizontally stretch a strategy with factor $C\in(0,1]$, it is interesting that a proportional utility is guaranteed, as stated in the following proposition.
\begin{proposition}\label{prop:relationship between w and b under C}
    Given the value function $v_j(q)$ of designer $j$ and the opponent's strategy $x_{-j}(h)$, the following propositions hold:
    \begin{enumerate}
        \item If $x_j(h)$ $C$-dominates $\hat{x}_j(h)$ for some $C\in(0,1]$, then $R^{Worst}(x_j(h))$  $\geq C \cdot R^{Worst}(\hat{x}_j(h))$ and  $R^{Best}(x_j(h))\geq C \cdot R^{Best}(\hat{x}_j(h))$ hold;
        \item If $x_j(h)$ $C$-strongly-dominates $\hat{x}_j(h)$ for some $C\in(0,1]$, then $R^{Worst}(x_j(h))\geq C \cdot R^{Best}(\hat{x}_j(h))$ holds.
    \end{enumerate}
\end{proposition}

In addition, we obtain lower bounds on the utility in both the worst and the best equilibrium, which is helpful in following discussions. 
\begin{proposition}\label{proposition:best-response-trivial-lowerbound}
For any $v_j(q)$ and $x_{-j}(h)$, the following statements hold.
\begin{enumerate}
    \item If $x_j(h)$ $C$-dominates $x_{-j}(h)$, then $R^{Best}(x_j(h)|x_{-j}(h),v_j(q))\geq \frac{C}{C+1}\int_0^1v_j(q)dq$;
    \item If $x_j(h)$ $C$-strongly-dominates $x_{-j}(h)$, then $R^{Worst}(x_j(h)|x_{-j}(h),v_j(q))\geq \frac{C}{C+1}\int_0^1v_j(q)dq$;
    \item There is $x_j(h)\in\mathcal F_{t_j}$ such that $R^{Best}(x_j(h)|x_{-j}(h),v_j(q))\geq\frac{t_j}{t_j+t_{-j}}\int_0^1v_j(q)dq$.
\end{enumerate}
\end{proposition}
\vspace*{-10pt}
\section{Best Response of Designer}
\label{section:best_response}

With the preparatory results in Section \ref{sec:relax and simp}, we investigate the strategy of a contest designer starting from this section. 
Two kinds of situations are studied: last-mover strategy (Section \ref{section:best_response}) and first-mover strategy (Section \ref{section:safetylevel}). In this section, we concentrate on the optimal last-mover strategy (best response), i.e., a designer decides her strategy after witnessing all other designers' strategies, aiming to maximize her utility. 


First, 
we specify the definition of the best-response problem, where the possible ambiguity caused by the non-uniqueness of the contestants' cumulative equilibrium behavior is eliminated through some technical analysis. Then, we prove the NP-hardness of finding the best last-mover strategy by reduction from a weighted knapsack problem. Finally, we propose an FPTAS to output an $\epsilon$-approximate best response.

\subsection{Specification for Best Response Problem} 


As discussed in Subsection \ref{subsec:charac best and worst equ}, the GCEB of contestants can be non-unique. To define the best-response problem precisely, we must specify under which equilibrium behavior we are trying to optimize designer $j$'s utility.

Interestingly, the supremum of $R^{Best}(x_j(h))$
\footnote{In the rest part of this paper, we use the abbreviation $R^{Best}(x_j(h))$ and $R^{Worst}(x_j(h))$ to represent the best utility and the worst utility of designer of designer $j$, when there is no ambiguity.} 
is exactly equal to the supremum of  $R^{Worst}(x_j(h))$. It means that whichever measure we choose, the utility of the best response will be the same. We present this result in Theorem \ref{theorem:best-response-worst-and-best-equivalent}. 

\begin{theorem}\label{theorem:best-response-worst-and-best-equivalent}
Given any budget $t_j>0$ and the opponent's strategy $x_{-j}(h)$, for the designer $j$, $\sup_{x_j(h)\in\mathcal F_{t_j}}R^{Best}(x_j(h)$ $|x_{-j}(h),v_j(q))=\sup_{x_j(h)\in\mathcal F_{t_j}}R^{Worst}(x_j(h)|x_{-j}(h),v_j(q))$.
\end{theorem}
Here we present the proof of \Cref{theorem:best-response-worst-and-best-equivalent}, which is an example of the horizontal stretching method.
\begin{proof}
Let $R^*=\sup_{x_j(h)\in\mathcal F_{t_j}}R^{Best}(x_j(h))$.
For any $\epsilon>0$, there is some $\hat{x}_j(h)\in\mathcal F_{t_j}$ such that $R^{Best}(\hat{x}_j(h))\geq (1-\epsilon)R^*$. Then, we can construct another interim allocation function $\bar{x}_j(h)=\hat{x}_j(h/(1-\epsilon))+\epsilon t_j$ for any $h\in[0,1-\epsilon]$, and $\bar{x}_j(h)=\epsilon t_j$ for $h\in(1-\epsilon,1]$. One can check that $\int_0^1\bar{x}_j(h)dh= (1-\epsilon)\int_0^1\hat{x}_j(h)dh+\epsilon t_j\leq t_j$ (so $\bar{x}_j(h)\in\mathcal F_{t_j}$), and $\bar{x}_j(h)$ $(1-\epsilon)$-strongly-dominates $\hat{x}_j(h)$. By Proposition \ref{prop:relationship between w and b under C}, we have $R^{Worst}(\bar{x}_j(h))\geq (1-\epsilon)R^{Best}(\hat{x}_j(h))\geq (1-\epsilon)^2R^*$.
Since $\epsilon$ can be arbitrarily small, we get the desired result. 
\end{proof}

Theorem \ref{theorem:best-response-worst-and-best-equivalent} unifies different measures of utility, and guarantees that, when we search for the best response of designer $j$, we only need to care about the value of $R^{Best}(x_j(h))$. Therefore, the best response problem we consider in this section can be specified as: Given $x_{-j}(h)$, find the best $x_j(h)\in\mathcal F_{t_j}$ for designer $j$, which maximizes the expected utility under the best equilibrium of contestants, i.e., $R^{Best}(x_j(h)|x_{-j}(h),v_j(q))$.

\subsection{FPTAS for Best Response Problem}
In the last subsection, we eliminate the ambiguity and specified the problem. However, it is unfortunate that the above problem of finding the best-response is NP-hard.

\begin{theorem}\label{theorem:NP-hard}
Given $v_j(q)$, $x_{-j}(h)$ and $t_j>0$, it is NP-hard to calculate $j$'s best response.
\end{theorem}

Since the best response problem is NP-hard, we turn to designing the algorithm to output the approximate best response. We first come up with the definition of an $\epsilon$-best response. 
\begin{definition}
    Given the value function $v_j(q)$ and the strategy of the opponent $x_{-j}(h)$, we say that, for any $\epsilon>0$, $\hat{x}_j(h) \in \mathcal{F}_{t_j}$ is an $\epsilon$-best response, if 
    \begin{align*}
        R^{Best}(\hat{x}_j(h)|x_{-j}(h),v_j(q)) \geq (1- \epsilon)\sup_{x_j(h)\in\mathcal F_{t_j}}R^{Best}(x_j(h)|x_{-j}(h),v_j(q)).
    \end{align*}
\end{definition}
With the help of dynamic programming, we can design a fully polynomial time approximation scheme to find a $\epsilon$-best response. Due to space limitations, we put \Cref{alg:FPTAS} into appendix.
\begin{theorem}\label{theorem:FPTAS}
Given $x_{-j}(h)$, $v_j(q)$ and $t_j>0$, define $V(q)=\int_0^qv_j(u)du$,
 $V^{-1}(v)=\min\{q\in[0,1]:V(q)\geq v\}$, and $x_{-j}^{-1}(x)=\sup\{h\in[0,1]:x_{-j}(h)\geq x\}$. Given oracle access to those functions, there is an FPTAS to calculate a $\epsilon$-best response of designer $j$. 
\end{theorem}
\noindent \emph{Proof sketch.} Inspired by the classic FPTAS of knapsack problems, we design the FPTAS for the best-response problem based on dynamic programming. The main idea is that, although $\mathcal F_{t_j}$,the space of all possible strategies, is a continuous space, we can always discretize a good strategy $x_j(h)$ to a piecewise-constant function $\bar{x}_j(h)$, with some quantization constraints on the utility contribution and the budget requirement on each constant interval of $\bar{x}_j(h)$. Comparing with $x_j(h)$, $\bar{x}_j(h)$ needs $O(\epsilon)$ more budget and gets $O(\epsilon)$ less utility. We design a DP procedure to find the optimal strategy in the discrete space consisting of all the strategies under the quantization constraints, which is never worse than $\bar{x}_j(h)$, and thus ensures a $1-O(\epsilon)$ fraction of the best response utility. Finally, we deal with the $O(\epsilon)$-shortage of budget by horizontal stretching, and then we get a strategy which fits in the $t_j$ budget constraint, and is an $O(\epsilon)$-best response.

\section{Safety Level of  Designer}
\label{section:safetylevel}

In the previous section, we investigate the best response strategy of a designer, when she observes all other designers' prize structures clearly. In this section, we consider the first-mover strategy, i.e., making a decision without any knowledge of other designers' strategies. Similar with the simplification in Section \ref{sec:relax and simp}, we still focus on the situation of two contests, without loss of generality. 
Due to the uncertainty on others' strategies, we introduce the concept of \emph{Safety Level} to evaluate the performance of a first-mover strategy. 

\begin{definition}[Safety Level]
Given a designer $j$'s value function $v_j(q)$ and the budget of the opponent $t_{-j}$, for an interim allocation function $x_j(h)$, the safety level of $x_j(h)$ is defined as
\begin{align*}
    \mathrm{SL}_{t_{-j}, v_j(q)}(x_j(h))=\min_{x_{-j}(h) \in \mathcal F_{t_{-j}}}R^{\text{Worst}}(x_j(h)|x_{-j}(h),v_j(q)),
\end{align*}
\end{definition}

Generally speaking, the safety level represents the worst performance of one strategy and provides a lower bound on utility that one strategy can realize. In the worst case, the opponent can observe the first mover's strategy, and adaptively select a strategy which minimize the first mover's utility. Suppose the opponent produces a strategy by the horizontal stretching method, so that $x_{-j}(h)$ $\frac{t_{-j}}{t_j}$-dominates $x_j(h)$. This gives an upper bound on the safety level of any interim allocation function, as shown in the following proposition.  
\begin{proposition}\label{proposition:safetylevel-trivial-upper-bound}
For any value function $v_j(q)$ of designer $j$, for any $x_j(h)\in\mathcal F_{t_j}$, it holds that $\mathrm{SL}_{t_{-j},v_j(q)}(x_j(h))\leq\frac{t_j}{t_j+t_{-j}}\int_0^1v_j(q)dq$.
\end{proposition}

With this proposition, it is reasonable to exploit $\frac{t_j}{t_j+t_{-j}}\int_0^1v_j(q)dq$ as a benchmark of the safety level. Specifically, given $v_j(q)$, $t_j$ and $t_{-j}$, we say $x_j(h)\in\mathcal F_{t_j}$ is \emph{$C$-competitive}, if $\mathrm{SL}_{t_{-j},v_j(q)}(x_j(h))\geq \frac1C\cdot\frac{t_j}{t_j+t_{-j}}\int_0^1v_j(q)dq$. We will use the concept of $C$-competitive to appraise the performance of a strategy. 

Intuitively, it is harder to decide the strategy without knowing the opponent's strategy. The following example shows indeed that a class of interim allocation functions cannot guarantee a constant-competitive safety level, which suggests that it is non-trivial to find a constant-competitive strategy.

\begin{example}\label{example:simple-threshold-safetylevel}
Suppose $v_j(q)=\begin{cases}M,&\text{if }q\in[0,\frac1M],\\\frac1q,&\text{if }q\in(\frac1M,1]\end{cases}$, where $M$ is a big number, then $\int_0^1 v_j(q)dq=1+\ln M$ and $\max_{q\in[0,1]}qv_j(q)=1$. 
In this example, we consider a class of interim allocation functions, called \emph{simple threshold allocation function}, and demonstrate its property in Lemma \ref{lemma:example-simple-threshold-allocation}. 

\begin{lemma}\label{lemma:example-simple-threshold-allocation}
Given $t_j>0$, we say $x_j(h)\in \mathcal F_{t_j}$ is a simple threshold allocation function, if there is $r\in(0,1]$, such that $x_j(h)=\begin{cases}\frac{t_j}{r},&\text{if }h\in[0,r],\\0,&\text{if }h\in(r,1]\end{cases}$.
For any simple threshold allocation function $x_j(h)\in\mathcal F_{t_j}$, $\mathrm{SL}_{t_{-j};v_j(q)}(x_j(h))\leq\frac{t_j}{t_{-j}}\max_{q\in[0,1]}q v_j(q)$.
\end{lemma}
By \Cref{lemma:example-simple-threshold-allocation}, the safety level of any simple threshold allocation function $x_j(h)$ is bounded by
\begin{align*}
    \mathrm{SL}_{t_{-j};v_j(q)}(x_j(h))\leq\frac{t_j}{t_{-j}}\max_{q\in[0,1]}q v_j(q)=\frac{t_j}{t_{-j}}.
\end{align*}
However, comparing it with the benchmark, we get
\begin{align*}
    \frac{\mathrm{SL}_{t_{-j};v_j(q)}(x_j(h))}{\frac{t_j}{t_j+t_{-j}}\int_0^1v_j(q)dq}\leq \frac{t_j+t_{-j}}{t_{-j}(1+\ln M)}.
\end{align*}
It can be arbitrarily small when $M\to\infty$, which means that the class of simple threshold allocation functions does not guarantee a constant competitiveness. 
This result can be extended to show that a much broader function class also cannot guarantee a constant-competitive safety level. Due to space limitations, we demonstrate this in \Cref{example:extend-simple-threshold}.
\end{example}


In the following theorem we show that, for any value function $v_j(q)$ and any budget constraints $t_j,t_{-j}$, we can construct a piecewise constant function\footnote{We sometimes call such a monotone piecewise fucntion a staircase function.} $x_j(h)\in\mathcal F_{t_j}$, which has a 16-competitive safety level. The idea is, we carefully design the width and height of each stair\footnote{Here a stair denotes the interval on which $x_j(h)$ is a certain constant.} aiming for a certain value level. The opponent must invest a relatively large amount to defeat each stair, i.e., to prevent this stair from gaining enough value over the desired value level. However, the total budget of the opponent is limited. Therefore it can be guaranteed that some stairs success in getting the desired value, i.e., obtaining a good utility in total.
\begin{theorem}
\label{theorem:16-competitive-safetylevel}
For the designer $j$, given any $v_j(q)$, $t_j$ and $t_{-j}$, there exist a strategy $x_j(h)\in\mathcal F_{t_j}$ such that $\mathrm{SL}_{t_{-j},v_j(q)}(x_j(h))\geq \frac1{16}\frac{t_j}{t_j+t_{-j}}\int_0^1v(q)dq$.
\end{theorem}

\section{Rounding the Interim Allocation Functions}
\label{section:implement}

In \Cref{section:best_response} and \Cref{section:safetylevel}, we discuss the strategies
under the generalized model with relaxed space of interim allocation functions, as described in \Cref{sec:relax and simp}. This relaxation allows our analysis to be more convenient and flexible. For example, 
we can easily apply the horizontal stretching method to an interim allocation function, or construct interim allocation functions which are piecewise-constant, to solve the best-response and safety-level problem. Note that these results can even be derived without knowing $n$.

However, in practice, a contest designer must implement her strategy as a rank-by-skill allocation, setting the $n$ prizes, which is the original model, and an interim allocation function generally cannot be implemented directly. Luckily, we can show that, any interim allocation can be approximated closely by a rank-by-skill allocation, when $n$ is sufficiently large.

\begin{theorem}\label{theorem:approximating-interim-allocation-with-rank-by-skill}
Given the value function $v_j(q)$, the own budget $t_j$ and the opponent's budget $t_{-j}$, let $M=v_j(0)$ be the maximum value of $v_j(q)$, $K=\int_0^1v_j(q)dq$ be the expected value of $v_j(q)$, and $K^*=\frac{t_j}{t_j+t_{-j}}K$  which is the lower bound for best-response problem (\Cref{proposition:best-response-trivial-lowerbound}) and the upper bound for safety-level problem (\Cref{proposition:safetylevel-trivial-upper-bound}).

Fix some constant $D>1$. When $n$ is sufficiently large, for any $x_j(h)\in\mathcal F_{t_j}$, there is a rank-by-skill prize structure $\vec{w}=(w_1,w_2,\cdots,w_n)$ satisfying that $w_1\geq w_2\geq\cdots\geq w_n\geq 0$ and $\sum_{k=1}^nw_k\leq nt_j=T_j$. Moreover, for any $x_{-j}(h)\in\mathcal F_{t_{-j}}$, if the following inequality holds,
\begin{equation}
    R^{Best}(x_j(h)|x_{-j}(h),v_j(q))\geq K^*/D,\label{eq:rank-by-skill-approximating-condition}
\end{equation}
then we have 
$$R^{Worst}(x_{\vec{w}}(h)|x_{-j},v_j(q))\geq (1-r(n)) R^{Best}(x_j(h)|x_{-j},v_j(q)),$$
where $x_{\vec{w}}(h)=\sum_{k=1}^n w_{k}\binom{n-1}{k-1}h^{k-1}(1-h)^{n-k}$ and $$r(n)=O\left(\left(\frac{DM}{K^*}+\sqrt{\ln n}\right)n^{-\frac13}\right)=o(1).$$
\end{theorem}

\Cref{theorem:approximating-interim-allocation-with-rank-by-skill} tells us that, when the value function $v_j(q)$ and the budgets $t_j, t_{-j}$ are fixed, we get an error bound $r(n)$ which is independent of $x_j(h)$ and $x_{-j}(h)$, and tends to $0$ as $n$ tends to infinity. For any $x_j(h)$, we can construct a rank-by-skill prize structure $\vec{w}=(w_1,w_2,\cdots,w_n)$, such that for any  $x_{-j}(h)$, $x_{\vec{w}}(h)$ is an $(1-r(n))$-approximation of $x_j(h)$, with respect to the utility, whenever $x_j(h)$ gets an somewhat satisfactory result against $x_{-j}(h)$, i.e., whenever \Cref{eq:rank-by-skill-approximating-condition} holds.

For the best-response problem, when $x_{-j}(h)$ is known, it is fairly easy to find some $x_j(h)$ such that $R^{Best}(x_j(h)|x_{-j},v_j(q))\geq K^*$, therefore satisfying \Cref{eq:rank-by-skill-approximating-condition}, as shown in the proof of \Cref{proposition:best-response-trivial-lowerbound} using the horizontal stretching method. Obviously, the $\epsilon$-approximate solution found by the FPTAS (\Cref{theorem:FPTAS}) also satisfies \Cref{eq:rank-by-skill-approximating-condition}, for any $\epsilon$ that is not too big. Therefore, when $n$ is large, a designer can first find a satisfactory strategy $x_j(h)$, then get a rank-by-skill prize structure by applying \Cref{theorem:approximating-interim-allocation-with-rank-by-skill}, without significant loss on utility.

For the safety-level problem, since we have constructed a constant-competitive strategy $x_j(h)$ as shown in \Cref{theorem:16-competitive-safetylevel}, we can see that $x_j(h)$ satisfies \Cref{eq:rank-by-skill-approximating-condition}. By applying \Cref{theorem:approximating-interim-allocation-with-rank-by-skill}, we can get a rank-by-skill prize structure that is always a $(1-r(n))$-approximation and consequently has a constant-competitive safety level when $n$ is sufficiently large.

We remark that, designing the allocation rule of a contest in the space of rank-by-skill allocations seems more difficult than that in the generalized space of any interim allocation function, due to the indivisibility of a prize: the designer cannot hand out half a prize to half a contestant, to get half of the contestant's value. The following example shows this difference between the original setting and the generalized setting.
\begin{example}
Assume that $n$ is some positive integer more than $2$ and there are $m=2$ designers with $nt_1<1$ and $nt_2=n$. Suppose $v_j(q)=1$, $\forall q\in[0,1]$, $\forall j=1,2$. If we know that the prize structure of contest $2$ is $w_{2,k}=1$ for $k=1,\cdots,n$, now consider the best response of designer $1$.

In the original model where designer $1$ has to choose a rank-by-skill allocation $\vec{w}_1=(w_{1,1},\cdots,$  $w_{1,n})$. We have $\max_{j\in[0,1]}x_{1}(h)=x_{1}(0)=w_{1,1}\leq\sum_{k=1}^nw_{1,k}\leq nt_1<1$. However, it can be seen that $x_2(h)=\sum_{k=1}^n w_{2,k}\binom{n-1}{k-1}h^{k-1}(1-h)^{n-k}=1$ for any $h\in[0,1]$. This implies that $Q(1)=x_1^{-1}(1)+x_2^{-1}(1)=0+1=1$, so by \Cref{theorem:player-equilibrium-general-condition}, we have $H_1^*(1)\leq x_1^{-1}(1)=0$ for any cumulative equilibrium behavior $H^*=(H_1^*(q),H_2^*(q))$. In other words, designer $1$ cannot attract any contestant and definitely gets zero utility.

In contrast, in the generalized model where designer $1$ is allowed to choose any non-negative and non-increasing interim allocation function $x_1(h)\in\mathcal{F}_{t_1}$ and use the abstract allocation. 
She can choose $x_1(h)=\begin{cases}1+\epsilon,&h\in[0,\frac{t_1}{1+\epsilon}],\\0,&h>\frac{t_1}{1+\epsilon}\end{cases}$, where $\epsilon>0$ can be arbitrarily small, and gets a positive utility of $R^{Worst}(x_1(h)|x_2(h),v_1(q))=\frac{t_1}{1+\epsilon}$. 
\end{example}

This example also shows that, the ratio between the utility of best response in the generalized setting and that in the original setting can be arbitrarily large when $n$ is not sufficiently large.



\section{Conclusions and Future Works}
\label{section: conclusion}

In summary, this paper studies the competition among several contest designers. There are $n$ contestants, each of whom choose one contest to join, and try to maximize the expected prize she received. The designers design their prize structure to maximize the total value of participants. For the game of contestants, we give a complete characterization of symmetric Bayesian Nash equilibria. For the game of designers, we investigate the last-mover and first-mover strategy design in the relaxed space. We propose an FPTAS for the best-response problem and a construction of a first-mover strategy with $16$-competitive safety level. Finally we round a strategy back to an implementable prize structure approximately.

We propose three directions for future works that are worth consideration. The first direction is that, we only showed the existence of a first-mover strategy with $16$-competitive safety level comparing with the upper bound, and it remains unclear whether strategies with better safety level exists. The second direction is to investigate the equilibria of designers and characterize the equilibrium prize policies. The last one is to consider the model where the contestants can attend more than one contest and spilt effort over these contests.

\bibliographystyle{splncs04} 
\bibliography{main}

\begin{thebibliography}{10}
\providecommand{\url}[1]{\texttt{#1}}
\providecommand{\urlprefix}{URL }
\providecommand{\doi}[1]{https://doi.org/#1}

\bibitem{AV79}
Angluin, D., Valiant, L.G.: Fast probabilistic algorithms for hamiltonian
  circuits and matchings. J. Comput. Syst. Sci.  \textbf{18}(2),  155--193
  (1979)

\bibitem{AS09}
Archak, N., Sundararajan, A.: Optimal design of crowdsourcing contests. ICIS
  2009 proceedings p.~200 (2009)

\bibitem{AM09}
Azmat, G., M{\"o}ller, M.: Competition among contests. The RAND Journal of
  Economics  \textbf{40}(4),  743--768 (2009)

\bibitem{AM18}
Azmat, G., M{\"o}ller, M.: The distribution of talent across contests. The
  economic journal  \textbf{128}(609),  471--509 (2018)

\bibitem{BKL22}
Birmpas, G., Kovalchuk, L., Lazos, P., Oliynykov, R.: Parallel contests for
  crowdsourcing reviews: Existence and quality of equilibria. arXiv preprint
  arXiv:2202.04064  (2022)

\bibitem{B16}
B{\"u}y{\"u}kboyac{\i}, M.: A designer's choice between single-prize and
  parallel tournaments. Economic Inquiry  \textbf{54}(4),  1774--1789 (2016)

\bibitem{CHS19}
Chawla, S., Hartline, J.D., Sivan, B.: Optimal crowdsourcing contests. Games
  and Economic Behavior  \textbf{113},  80--96 (2019)

\bibitem{KG03}
Che, Y.K., Gale, I.: Optimal design of research contests. American Economic
  Review  \textbf{93}(3),  646--671 (2003)

\bibitem{C07}
Corch{\'o}n, L.C.: The theory of contests: a survey. Review of economic design
  \textbf{11}(2),  69--100 (2007)

\bibitem{DGLLL21}
Deng, X., Gafni, Y., Lavi, R., Lin, T., Ling, H.: From monopoly to competition:
  Optimal contests prevail. arXiv preprint arXiv:2107.13363  (2021)

\bibitem{DV09}
DiPalantino, D., Vojnovic, M.: Crowdsourcing and all-pay auctions. In:
  Proceedings of the 10th ACM conference on Electronic commerce. pp. 119--128
  (2009)

\bibitem{EGG21}
Elkind, E., Ghosh, A., Goldberg, P.: Contest design with threshold objectives.
  arXiv preprint arXiv:2109.03179  (2021)

\bibitem{EGG22}
Elkind, E., Ghosh, A., Goldberg, P.W.: Simultaneous contests with equal sharing
  allocation of prizes: Computational complexity and price of anarchy. In:
  International Symposium on Algorithmic Game Theory. pp. 133--150. Springer
  (2022)

\bibitem{GH12}
Ghosh, A., Hummel, P.: Implementing optimal outcomes in social computing: A
  game-theoretic approach. In: Proceedings of the 21st international conference
  on World Wide Web. pp. 539--548 (2012)

\bibitem{GK16}
Ghosh, A., Kleinberg, R.: Optimal contest design for simple agents. ACM
  Transactions on Economics and Computation (TEAC)  \textbf{4}(4),  1--41
  (2016)

\bibitem{GH88}
Glazer, A., Hassin, R.: Optimal contests. Economic Inquiry  \textbf{26}(1),
  133--143 (1988)

\bibitem{HHKK18}
Hafalir, I.E., Hakimov, R., K{\"u}bler, D., Kurino, M.: College admissions with
  entrance exams: Centralized versus decentralized. Journal of Economic Theory
  \textbf{176},  886--934 (2018)

\bibitem{JSY20}
Juang, W.T., Sun, G.Z., Yuan, K.C.: A model of parallel contests. International
  Journal of Game Theory  \textbf{49}(2),  651--672 (2020)

\bibitem{KC18}
K{\"o}rpeo{\u{g}}lu, E., Cho, S.H.: Incentives in contests with heterogeneous
  solvers. Management Science  \textbf{64}(6),  2709--2715 (2018)

\bibitem{KK22}
K{\"o}rpeo{\u{g}}lu, E., Korpeoglu, C.G., Hafal{\i}r, {\.I}.E.: Parallel
  innovation contests. Operations Research  (2022)

\bibitem{LS21}
Lavi, R., Shiran-Shvarzbard, O.: Competition among contests: a safety level
  analysis. In: Proceedings of the Twenty-Ninth International Conference on
  International Joint Conferences on Artificial Intelligence. pp. 378--385
  (2021)

\bibitem{LS18}
Levy, P., Sarne, D.: Understanding over participation in simple contests. In:
  Proceedings of the AAAI Conference on Artificial Intelligence. vol.~32 (2018)

\bibitem{LSA19}
Levy, P., Sarne, D., Aumann, Y.: Selective information disclosure in contests.
  In: Proceedings of the 18th International Conference on Autonomous Agents and
  MultiAgent Systems. pp. 2093--2095 (2019)

\bibitem{LSH20}
Levy, P., Sarne, D., Habani, M.: Simple contest enhancers. In: 2020
  IEEE/WIC/ACM International Joint Conference on Web Intelligence and
  Intelligent Agent Technology (WI-IAT). pp. 463--469. IEEE (2020)

\bibitem{LSR17}
Levy, P., Sarne, D., Rochlin, I.: Contest design with uncertain performance and
  costly participation. In: IJCAI. pp. 302--309 (2017)

\bibitem{MS01}
Moldovanu, B., Sela, A.: The optimal allocation of prizes in contests. American
  Economic Review  \textbf{91}(3),  542--558 (2001)

\bibitem{MS06}
Moldovanu, B., Sela, A.: Contest architecture. Journal of Economic Theory
  \textbf{126}(1),  70--96 (2006)

\bibitem{MSV18}
Morgan, J., Sisak, D., V{\'a}rdy, F.: The ponds dilemma. The Economic Journal
  \textbf{128}(611),  1634--1682 (2018)

\bibitem{RF88}
Royden, H.L., Fitzpatrick, P.: Real analysis, vol.~32. Macmillan New York
  (1988)

\bibitem{S20}
Segev, E.: Crowdsourcing contests. European Journal of Operational Research
  \textbf{281}(2),  241--255 (2020)

\bibitem{T02}
Tennenholtz, M.: Competitive safety analysis: Robust decision-making in
  multi-agent systems. Journal of Artificial Intelligence Research
  \textbf{17},  363--378 (2002)

\bibitem{TX08}
Terwiesch, C., Xu, Y.: Innovation contests, open innovation, and multiagent
  problem solving. Management science  \textbf{54}(9),  1529--1543 (2008)

\bibitem{T01}
Tullock, G.: Efficient rent seeking. In: Efficient rent-seeking, pp. 3--16.
  Springer (2001)

\end{thebibliography}

\newpage
\appendix
\section*{Appendix}
This appendix consists of five parts 
(Appendix \ref{app:sec3}, \ref{app:sec4}, \ref{app:sec5}, \ref{app:sec6} and \ref{app:sec7}) listing the missing proofs of Section \ref{section:players_equilibrium}, \ref{sec:relax and simp}, \ref{section:best_response}, \ref{section:safetylevel} and \ref{section:implement}, respectively.

\section{Missing Proofs in Section \ref{section:players_equilibrium}}
\label{app:sec3}

\subsection{Proof of \Cref{proposition:differentiating-cumulative-behavior}}
\begin{proof}
First, we claim that each $H_j(q)$ satisfies the Lipschitz condition, because, for any $q_1<q_2 \in [0, 1]$, we have $0\leq H_j(q_2)-H_j(q_1)\leq \sum_{j=1}^m(H_j(q_2)-H_j(q_1))=q_2-q_1$. A Lipschitz function is absolutely continuous, and the following result is well-known:
\begin{lemma}[\cite{RF88}]
If $f(x)$ is an absolutely continuous function on $[a,b]$, then $f(x)$ is differentiable almost everywhere, and its pointwise derivative $f'(x)$ is integrable on $(a,b)$, and moreover, for every $x\in[a,b]$, $f(x)=\int_a^xf'(t)dt$.
\end{lemma}

It is obvious that $E_j$ is a zero measure set. 
By the above lemma, for any $q\in[0,1]\setminus E$ and any $j\in[m]$, the pointwise derivative $H'_j(q)$ exists. Thus, it implies that  $\sum_{j=1}^m\pi_j(q)=\sum_{j=1}^mH_j'(q)=1$ (because $\sum_{j=1}^mH_j(q)=q$); if $q\in E$, then $\sum_{j=1}^m\pi_j(q)=\frac{\left|\arg\max_{j\in[m]}x_{j}(H_{j}(q)+0)\right|}{\left|\arg\max_{j\in[m]}x_{j}(H_{j}(q)+0)\right|}=1$. Therefore, we show $\vec{\pi}\in \mathcal S$. Furthermore, since $E$ is a zero-measure set, we have $\int_0^q\pi_j(u)du=\int_0^qH'_j(u)du=H_j(q)$ for any $q\in[0,1]$. In summary, $\vec{H}(q)$ is the cumulation of $\vec{\pi}(q)$.
\qed
\end{proof}

\subsection{Proof of \Cref{theorem:rankbyskill-player-equilibrium-general-condition} and \Cref{theorem:player-equilibrium-general-condition}}
One can see that \Cref{theorem:rankbyskill-player-equilibrium-general-condition} is an immediate corollary of \Cref{theorem:player-equilibrium-general-condition} (stated in \Cref{sec:relax and simp}), so we only need to prove \Cref{theorem:player-equilibrium-general-condition}.
\begin{proof}

We first prove the necessity by contradiction. Assume that $\vec{H}^*(q)$ is a GCEB, and it is the cumulation of some GSBNE $\vec{\pi}^*\in \mathcal S$.

Suppose, for contradiction, that for some $X$, there is some $j_0\in[m]$, $q\in[q_1,\min(q_2,1)]$ that $H_{j_0}^*(q)<x_{j_0}^{-1}(X+0)$. Combining $q\geq q_1=Q(X+0)=\sum_{j=1}^mx_{j}^{-1}(X+0)$ and $\sum_{j=1}^mH_{j_0}^*(q)=q$, we have $\sum_{j'\neq j_0}H_{j'}^*(q)=q-H_{j_0}^*(q)>\sum_{j=1}^mx_{j}^{-1}(X+0)-x_{j_0}^{-1}(X+0)=\sum_{j'\neq j_0}x_{j'}^{-1}(X+0)$. Thus there must exist $j_1\neq j_0$ such that $H^*_{j_1}(q)>x_{j_1}^{-1}(X+0)$.

From $H_{j_0}^*(q)<x_{j_0}^{-1}(X+0)$ we get that there is $X'>X$, such that $H_{j_0}^*(q)<x_{j_0}^{-1}(X')$, which means that $x_{j_0}(H_{j_0}^*(q)+0)\geq X'>X$. On the other hand, since $H^*_{j_1}(q)=\int_0^q\pi^*_{j_1}(u)du$, and $H^*_{j_1}(q)>x_{j_1}^{-1}(X+0)\geq 0$, there exists $q'\in(0,q)$ such that $\pi^*_{j_1}(q')>0$ and $H^*_{j_1}(q')>x_{j_1}^{-1}(X+0)$, then $x_{j_1}(H^*_{j_1}(q')+0)\leq X$. Since $x_{j_0}(H^*_{j_0}(q')+0)\geq x_{j_0}(H^*_{j_0}(q)+0)>X$, we have $x_{j_1}(H^*_{j_1}(q')+0)\leq X<x_{j_0}(H^*_{j_0}(q')+0)$, which contradicts with \Cref{definition:contestants-generalized-symmetric-equilibrium}.

Similarly, suppose that for some $X$, there is some $j_0\in[m]$, $q\in[q_1,\min(q_2,1)]$ that $H_{j_0}^*(q)>x_{j_0}^{-1}(X)$. Since $\sum_{j=1}^mH_{j_0}^*(q)=q\leq q_2=Q(X)=\sum_{j=1}^mx_{j}^{-1}(X)$, we have $\sum_{j'\neq j_0}H_{j'}^*(q)=q-H_{j_0}^*(q)<\sum_{j'\neq j_0}x_{j'}^{-1}(X)$. Therefore there exists $j_1\neq j_0$ such that $H^*_{j_1}(q)<x_{j_1}^{-1}(X)$. This implies that $x_{j_1}(H^*_{j_1}(q)+0)\geq X$.

Since $H_{j_0}^*(q)>x_{j_0}^{-1}(X)\geq 0$, there is $q'\in(0,1)$ such that $\pi^*_{j_0}(q')>0$ and $H_{j_0}^*(q')>x_{j_0}^{-1}(X)$, then $x_{j_0}(H_{j_0}^*(q')+0)<X$. On the other hand $x_{j_1}(H^*_{j_1}(q')+0)\geq x_{j_1}(H^*_{j_1}(q)+0)\geq X$, therefore we obtain $x_{j_0}(H_{j_0}^*(q')+0)<X\leq x_{j_1}(H^*_{j_1}(q')+0)$, which contradicts with \Cref{definition:contestants-generalized-symmetric-equilibrium}.

Now we prove the sufficiency. Assume the condition holds for $\vec{H}^*$, we construct $\vec{\pi}^*(q)$ as in \Cref{proposition:differentiating-cumulative-behavior}, and show that $\vec{\pi}^*(q)$ satisfies \Cref{definition:contestants-generalized-symmetric-equilibrium}, and is thus a generalized symmetric Bayesian Nash equilibrium.
For every $j\in[m]$, for any $q\in[0,1]$, if $\pi^*_j(q)>0$, we can prove that $x_j(H_j^*(q)+0)=\max_{j'\in[m]}x_{j'}(H_{j'}^*(q)+0)$.

Consider two cases: Either $q\in E$, or $q \notin E$. If $q\in E$, by the construction of $\pi^*_j(q)$ we have known that $\pi^*_j(q)>0$ implies $j\in\arg\max_{j'\in[m]}x_{j'}(H_{j'}(q)+0)$.

If $q\notin E$, then $q\in(0,1)$, and $\pi^*_j(q) = (H_j^{*})'(q)$. Suppose that $\vec{\pi}^*(q)$ is not a GSBNE, then there exists a  $j'\neq j$ such that $x_{j'}(H_{j'}^*(q)+0)>x_j(H_{j}^*(q)+0)$, so there is $\delta>0$ such that $x_{j'}(H_{j'}^*(q)+\delta)>x_j(H_{j}^*(q)+0)$. Set $q'=q+\frac12\delta$, we get $H_{j'}^*(q')\leq H_{j'}^*(q)+\frac12\delta<H_{j'}^*(q)+\delta$. Therefore $x_j^{-1}(x_{j'}(H_{j'}^*(q')+0))\leq  x_j^{-1}(x_{j'}(H_{j'}^*(q)+\delta))$.

Since $x_{j'}(H_{j'}^*(q)+\delta)>x_j(H_{j}^*(q)+0)$, we have $x_j^{-1}(x_{j'}(H_{j'}^*(q)+\delta))\leq H_{j}^*(q)$, so $x_j^{-1}(x_{j'}(H_{j'}^*(q')+0))\leq x_j^{-1}(x_{j'}(H_{j'}^*(q)+\delta))\leq H_{j}^*(q)$.

Let $X'=Q^{-1}(q')$, then $Q(X+0)\leq q'\leq Q(X)$. By assumption we obtain $H_{j'}^*(q')\geq x_{j'}^{-1}(X'+0)=\sup\{h\in[0,1]:x_{j'}(h)>X'\}$. Thus we have $x_{j'}(H_{j'}^*(q')+0)\leq X'$. Take function $x_j^{-1}(\cdot)$ on the both sides, we get $x_j^{-1}(X')\leq x_j^{-1}(x_{j'}(H_{j'}^*(q')+0))\leq H_{j}^*(q)$. Again by assumption we have $H_{j}^*(q')\leq x_{j}^{-1}(X')$, therefore $H_{j}^*(q')\leq H_{j}^*(q)$. Since $q'>q$ and $H_j^*(q)$ is non-decreasing, it can only be a constant on $(q,q')$, contradicting with $H_j^{*\prime}(q)>0$.

Thus we have $x_j(H_{j}^*(q)+0)=\max_{j'\in[m]}x_{j'}(H_{j'}^*(q)+0)$ whenever $\pi^*_j(q)>0$, so $\vec{\pi}^*$ is a GSBNE, and $\vec{H}^*$ is a GCEB.

\qed
\end{proof}


\section{Missing Proofs in Section \ref{sec:relax and simp}}
\label{app:sec4}

\subsection{Proof of \Cref{proposition:m-contest-equivalent-with-two-contest}}
\begin{proof}
We declare some symbols for the new GoC game with two contests. Let $\tilde{m}=2$, $\tilde{x_1}(h)=x_j(h),\tilde{x_2}(h)=x_{-j}(h)$, and
$\tilde{Q}(x)=\sum_{j'=1}^{\tilde{m}}\tilde{x_{j'}}^{-1}(x)=x_j^{-1}(x)+x_{-j}^{-1}(x)$ and $\tilde{H_1}(q)=H_j(q),\tilde{H_2}(q)=\sum_{j'\neq j}H_{j'}(q)=q-H_j(q)$.
To prove that $\tilde{H}=(\tilde{H_1},\tilde{H_2})\in\mathcal H^*(\tilde{x_1}(h),\tilde{x_2}(h))$, we only need to check the condition \Cref{theorem:player-equilibrium-general-condition}.
Since $\tilde{Q}(x)=Q(x)$ still holds, we know $D_{\tilde{Q}(x)}=D_{Q(x)}$, and $\underline{X},\overline{X}$ are unchanged.

For any $X\in[\underline{X},\overline{X}]$, let $q_1=\tilde{Q}(X+0), q_2=\tilde{Q}(X)$. Then we have for any $q\in[q_1,\min(1,q_2)]$,
\begin{align*}
    & \tilde{H_1}(q)=H_j(q)\in[x_j^{-1}(X+0),x_j^{-1}(X)]=[\tilde{x_1}^{-1}(X+0),\tilde{x_1}^{-1}(X)] \text{,and}\\
    & \tilde{H_2}(q)=\sum_{j'\neq j}H_{j'}(q_1)\in[\sum_{j'\neq j}x_{j'}^{-1}(X+0),\sum_{j'\neq j}x_{j'}^{-1}(X)]=[\tilde{x_2}^{-1}(X+0),\tilde{x_2}^{-1}(X)].
\end{align*}

Therefore, by \Cref{theorem:player-equilibrium-general-condition}, we get $\tilde{H}(q)=(\tilde{H_1}(q),\tilde{H_2}(q))\in\mathcal H^*(x_j(h),x_{-j}(h))$.
\qed
\end{proof}

\subsection{A Technical Lemma}
\begin{lemma}
\label{lemma:integral-compare}
If $f(x)$ is non-negative and non-increasing on $[0,1]$, and $g_1(x)$, $g_2(x)$ are continuous non-decreasing functions on $[0,1]$, such that $g_1(0)=g_2(0)=0$, and for any $x\in[0,1]$, $g_1(x)\leq g_2(x)$. Then $\int_0^1f(x)dg_1(x)\leq \int_0^1f(x)dg_2(x)$.
\end{lemma}
\begin{proof}
For any partition $a=x_0\leq x_1\leq x_2\leq\cdots\leq x_n=b$ of $[0,1]$, and any choice $\{\xi_1,\xi_2,\cdots,\xi_n\}$ such that $\xi_i\in[x_{i-1},x_i]$, we have $$\sum_{i=1}^nf(\xi_i)(g_1(x_{i})-g_1(x_{i-1}))=\sum_{i=1}^{n-1}g_1(x_{i})(f(\xi_i)-f(\xi_{i+1}))+f(\xi_n)g_1(x_n),$$
and $$\sum_{i=1}^nf(\xi_i)(g_2(x_{i})-g_2(x_{i-1}))=\sum_{i=1}^{n-1}g_2(x_{i})(f(\xi_i)-f(\xi_{i+1}))+f(\xi_n)g_2(x_n).$$
By assumption we have $f(\xi_i)-f(\xi_{i+1})\geq 0$ , $f(\xi_n)\geq 0$, and $g_1(x_i)\leq g_2(x_i)$. Combining with the above two equations, we get $$\sum_{i=1}^nf(\xi_i)(g_1(x_{i})-g_1(x_{i-1}))\leq \sum_{i=1}^nf(\xi_i)(g_2(x_{i})-g_2(x_{i-1})).$$
By the definition of Riemann-Stieltjes integral, we have the result,  $\int_0^1f(x)dg_1(x)$ $\leq \int_0^1f(x)dg_2(x)$.
\qed
\end{proof}

\subsection{Proof of \Cref{proposition:best-response-worst-and-best-equilibrium}}
\begin{proof}
We first prove that $H_j^{Worst}(q)$ and $H_j^{Best}(q)$ belong to $\mathcal H^*_j(x_j(h),x_{-j}(h))$.

We first check that $(H_j^{Worst}(q),q-H_j^{Worst}(q))\in\mathcal H^*(x_j(h),x_{-j}(H))$.
$\forall q\in [0,1]$, let $X=Q^{-1}(q)$, then $Q(X+0)\leq q\leq Q(X)$. Since $H_j^{Worst}(q)=\max\{x_j^{-1}(X+0), q-x_{-j}^{-1}(X)\}$, we have $H_j^{Worst}(q)\geq x_j^{-1}(X+0)$. Moreover, combining $x_j^{-1}(X+0)\leq x_j^{-1}(X)$ and $q-x_{-j}^{-1}(X)\leq Q(X)-x_{-j}^{-1}(X)=x_j^{-1}(X)$, we have $H_j^{Worst}(q)\leq x_j^{-1}(X)$. So we have $x_j^{-1}(X+0)\leq H_j^{Worst}(q)\leq x_j^{-1}(X)$.

For $q-H_j^{Worst}(q)$, since $q-H_j^{Worst}(q)=\min\{q-x_j^{-1}(X+0), x_{-j}^{-1}(X)\}$, we have $q-H_j^{Worst}(q)\leq x_{-j}^{-1}(X)$. Combining $q-x_j^{-1}(X+0)\geq Q(X+0)-x_j^{-1}(X+0)=x_{-j}^{-1}(X+0)$ and $x_{-j}^{-1}(X)\geq x_{-j}^{-1}(X+0)$, we have $q-H_j^{Worst}(q)\geq x_{-j}^{-1}(X+0)$. So we have $x_{-j}^{-1}(X+0)\leq q-H_j^{Worst}(q)\leq x_{-j}^{-1}(X)$. Thus by \Cref{theorem:player-equilibrium-general-condition}, we have $(H_j^{Worst}(q),q-H_j^{Worst}(q))\in\mathcal H^*(x_j(h),x_{-j}(H))$.

Almost symmetrically, we can check that $x_j^{-1}(X+0)\leq H_j^{Best}(q)\leq x_j^{-1}(X)$ and $x_{-j}^{-1}(X+0)\leq q-H_j^{Best}(q)\leq x_{-j}^{-1}(X)$, we omit the detailed calculation. Again by \Cref{theorem:player-equilibrium-general-condition}, we have $(H_j^{Best}(q),q-H_j^{Best}(q))\in\mathcal H^*(x_j(h),x_{-j}(H))$.
Therefore we have that both $H_j^{Worst}(q)$ and $H_j^{Best}(q)$ belong to $\mathcal H^*_j(x_j(h),x_{-j}(h))$.

For any $H_j(q)\in\mathcal H^*_j(x_j(h),x_{-j}(h))$, for any $q\in[0,1]$, by \Cref{theorem:player-equilibrium-general-condition} we have $x_j^{-1}(X+0)\leq H_j(q)\leq x_j^{-1}(X)$ and $x_{-j}^{-1}(X+0)\leq q-H_j(q)\leq x_{-j}^{-1}(X)$. The latter inequality is equivalent with $q-x_{-j}^{-1}(X)\leq H_j(q)\leq q-x_{-j}^{-1}(X+0)$. Thus we have $\max\{x_j^{-1}(X+0),q-x_{-j}^{-1}(X)\}\leq H_j(q)\leq \min\{x_j^{-1}(X),q-x_{-j}^{-1}(X+0)\}$, i.e., $H_j^{Worst}(q)\leq H_j(q)\leq H_j^{Best}(q)$. By \Cref{lemma:integral-compare}, we obtain
$$\int_0^1v_j(q)dH_j^{Worst}(q)\leq\int_0^1v_j(q)dH_j(q)\leq\int_0^1v_j(q)dH_j^{Best}(q).$$

Since $H_j(q)$ can be arbitrarily chosen from $\mathcal H^*_j(x_j(h),x_{-j}(h))$, we have $$\int_0^1v_j(q)dH_j^{Worst}(q)=\min_{H_j(q)\in\mathcal H^*_j(x_j(h),x_{-j}(h))}\int_0^1v_j(q)dH_j(q)=R^{Worst}(x_j(h)),$$and
$$\int_0^1v_j(q)dH_j^{Best}(q)=\max_{H_j(q)\in\mathcal H^*_j(x_j(h),x_{-j}(h))}\int_0^1v_j(q)dH_j(q)=R^{Best}(x_j(h)).$$
\qed
\end{proof}

\subsection{Proof of Proposition \ref{prop:relationship between w and b under C}}

\begin{proof}
We first prove the statement 1. Let $H_j^{Worst}$ and $H_j^{Best}$ be the worst and the best behavior of contestants in $\mathcal H_j^*(x_j(h),x_{-j}(h))$ for designer $j$, and $\hat{H}_j^{Worst}$ and $\hat{H}_j^{Best}$ be the worst and the best behavior of contestants in $\mathcal H_j^*(\hat{x}_j(h),x_{-j}(h))$, as defined in \Cref{proposition:best-response-worst-and-best-equilibrium}.

Consider two special cases: the case that $C=1$, i.e. $x(h)\geq \hat{x}(h)$ for any $h\in[0,1]$; and the case that $x_j(Ch)=\hat{x}(h)$ for any $h\in[0,1]$.

Firstly we prove the case of $C=1$.
If $x_j(h)$ $1$-dominates $\hat{x}_j(h)$, then for any $h\in[0,1]$, $x_j(h)\geq \hat{x}_j(h)$.
Let $Q(x)=x_{j}^{-1}(x)+x_{-j}^{-1}(x),\hat{Q}(x)=\hat{x}_{j}^{-1}(x)+x_{-j}^{-1}(x)$. 

For any $X\in[\hat{Q}^{-1}(1),\hat{x}_j(0)]$, $x_j^{-1}(X)\geq \hat{x}_j^{-1}(X)$, and $Q(X)\geq \hat{Q}(X)$. For any $q\in[0,1]$, let $X=Q^{-1}(q)$, $\hat{X}=\hat{Q}^{-1}(q)$, then $\hat{X}\leq X$.

If $\hat{X}<X$, then $q-\hat{H}_j^{Worst}(q)\geq q-\hat{H}_j^{Best}(q)\geq x_{-j}^{-1}(\hat{X}+0)\geq x_{-j}^{-1}(X)\geq q-H_j^{Worst}(q)\geq q-H_j^{Best}(q)$, so $H_j^{Best}(q)\geq H_j^{Worst}(q)\geq \hat{H}_j^{Best}(q)\geq \hat{H}_j^{Worst}(q)$.

If $\hat{X}=X$, then
\begin{align*}
\hat{H}_j^{Best}(q)=&\min(q-x_{-j}^{-1}(X+0),\hat{x}_{j}^{-1}(X)),\\
H_j^{Best}(q)=&\min(q-x_{-j}^{-1}(X+0),x_{j}^{-1}(X)),\\ \hat{H}_j^{Worst}(q)=&\max(q-x_{-j}^{-1}(X),\hat{x}_{j}^{-1}(X+0)),\\ H_j^{Worst}(q)=&\max(q-x_{-j}^{-1}(X),x_{j}^{-1}(X+0)).
\end{align*}
Since $\hat{x}_{j}^{-1}(X)\leq x_{j}^{-1}(X)$ and $\hat{x}_{j}^{-1}(X+0)\leq x_{j}^{-1}(X+0)$, we have $\hat{H}_j^{Best}(q)\leq H_j^{Best}(q)$, and  $\hat{H}_j^{Worst}(q)\leq H_j^{Worst}(q)$.

Therefore $\forall q\in[0,1]$, $\hat{H}_j^{Best}(q)\leq H_j^{Best}(q)$, and  $\hat{H}_j^{Worst}(q)\leq H_j^{Worst}(q)$. By \Cref{lemma:integral-compare},  $\int_0^1v_j(q)d\hat{H}_j^{Best}(q)\leq \int_0^1v_j(q)dH_j^{Best}(q)$ and $\int_0^1v_j(q)d\hat{H}_j^{Worst}(q)\leq \int_0^1v_j(q)dH_j^{Worst}(q)$.

Next we prove the case that for some $C\in(0,1]$, $\forall h\in[0,1]$, $x_j(Ch)=\hat{x}(h)$.
We have $x_j^{-1}(x)=C\hat{x}_j^{-1}(x)$ for any $x$.
Let $Q(x)=x_{j}^{-1}(x)+x_{-j}^{-1}(x),\hat{Q}(x)=\hat{x}_{j}^{-1}(x)+x_{-j}^{-1}(x)$. 

For any $q\in[0,1]$, let $\hat{X}=\hat{Q}^{-1}(q)$, then $\hat{H}_j^{Best}(q)=\min(q-x_{-j}^{-1}(\hat{X}+0),\hat{x}_{j}^{-1}(\hat{X}))$. Let $q'=C\hat{H}_j^{Best}(q)+x_{-j}^{-1}(\hat{X}+0)$, then $q'\leq \hat{H}_j^{Best}(q)+x_{-j}^{-1}(\hat{X}+0)\leq q$. We have
\begin{align*}
&\hat{x}_{j}^{-1}(\hat{X}+0)\leq \hat{H}_j^{Best}(q)\leq \hat{x}_{j}^{-1}(\hat{X})\\
\implies& C\hat{x}_{j}^{-1}(\hat{X}+0)+x_{-j}^{-1}(\hat{X}+0)\leq q'\leq C\hat{x}_{j}^{-1}(\hat{X})+x_{-j}^{-1}(\hat{X}+0)\\
\implies& x_{j}^{-1}(\hat{X}+0)+x_{-j}^{-1}(\hat{X}+0)\leq q'\leq x_{j}^{-1}(\hat{X})+x_{-j}^{-1}(\hat{X})\\
\implies& Q(\hat{X}+0)\leq q'\leq Q(\hat{X}).
\end{align*}
Thus $H_j^{Best}(q')=\min(q'-x_{-j}^{-1}(\hat{X}+0),x_{j}^{-1}(\hat{X}))\leq q'-x_{-j}^{-1}(\hat{X}+0)=C\hat{H}_j^{Best}(q)$. Then $H_j^{Best}(q)\geq H_j^{Best}(q')\geq C\hat{H}_j^{Best}(q)$ since $q'\leq q$.

For any $q\in[0,1]$, let $\hat{X}=\hat{Q}^{-1}(q)$, then $\hat{H}_j^{Worst}(q)=\max(q-x_{-j}^{-1}(\hat{X}),\hat{x}_{j}^{-1}(\hat{X}+0))$.
When $q-x_{-j}^{-1}(\hat{X})\geq\hat{x}_{j}^{-1}(\hat{X}+0)$, $\hat{H}_j^{Worst}(q)=q-x_{-j}^{-1}(\hat{X})$. Let $q'=C\hat{H}_j^{Worst}(q)+x_{-j}^{-1}(\hat{X})$, then $q'\leq \hat{H}_j^{Worst}(q)+x_{-j}^{-1}(\hat{X})=q$. We have
\begin{align*}
&\hat{x}_{j}^{-1}(\hat{X}+0)\leq \hat{H}_j^{Worst}(q)\leq \hat{x}_{j}^{-1}(\hat{X})\\
\implies& C\hat{x}_{j}^{-1}(\hat{X}+0)+x_{-j}^{-1}(\hat{X}+0)\leq q'\leq C\hat{x}_{j}^{-1}(\hat{X})+x_{-j}^{-1}(\hat{X}+0)\\
\implies& x_{j}^{-1}(\hat{X}+0)+x_{-j}^{-1}(\hat{X}+0)\leq q'\leq x_{j}^{-1}(\hat{X})+x_{-j}^{-1}(\hat{X})\\
\implies& Q(\hat{X}+0)\leq q'\leq Q(\hat{X}),
\end{align*}
Thus $H_j^{Worst}(q')=\max(q'-x_{-j}^{-1}(\hat{X}),x_{j}^{-1}(\hat{X}+0))=\max(C\hat{H}_j^{Worst}(q),x_{j}^{-1}(\hat{X}+0))$. Since $\hat{H}_j^{Worst}(q)=q-x_{-j}^{-1}(\hat{X})\geq \hat{x}_{j}^{-1}(\hat{X}+0)$, we obtain $C\hat{H}_j^{Worst}(q)\geq x_{j}^{-1}(\hat{X}+0)$, and $H_j^{Worst}(q')=C\hat{H}_j^{Worst}(q)$. Then $H_j^{Worst}(q)\geq H_j^{Worst}(q')=C\hat{H}_j^{Worst}(q)$ since $q'\leq q$.

When $q-x_{-j}^{-1}(\hat{X})\leq\hat{x}_{j}^{-1}(\hat{X}+0)$, $\hat{H}_j^{Worst}(q)=\hat{x}_{j}^{-1}(\hat{X}+0)$. Let $q'=C\hat{H}_j^{Worst}(q)+x_{-j}^{-1}(\hat{X}+0)$, then $q'=x_{j}^{-1}(\hat{X}+0)+x_{-j}^{-1}(\hat{X}+0)=Q(\hat{X}+0)$. We have $Q(\hat{X}+0)=q'\leq C\hat{H}_j^{Worst}(q)+(q-\hat{H}_j^{Worst}(q))\leq q$, and $Q(\hat{X})\geq Q(\hat{X}+0)=q'$.

Thus, $H_j^{Worst}(q')=\max(q'-x_{-j}^{-1}(\hat{X}),x_{j}^{-1}(\hat{X}+0))\geq x_{j}^{-1}(\hat{X}+0)=C\hat{x}_{j}^{-1}(\hat{X}+0)=C\hat{H}_j^{Worst}(q)$. Then $H_j^{Worst}(q)\geq H_j^{Worst}(q')\geq C\hat{H}_j^{Worst}(q)$ since $q'\leq q$.

Now we have $H_j^{Best}(q)\geq C\hat{H}_j^{Best}(q)$ and $H_j^{Worst}(q)\geq C\hat{H}_j^{Worst}(q)$, by \Cref{lemma:integral-compare} we have
\begin{align*}\int_0^1v_j(q)dH_j^{Best}(q)&\geq\int_0^1v_j(q)d(C\hat{H}_j^{Best}(q))= C\int_0^1v_j(q)d\hat{H}_j^{Best}(q),\\ \int_0^1v_j(q)dH_j^{Worst}(q)&\geq\int_0^1v_j(q)d(C\hat{H}_j^{Worst}(q))= C\int_0^1v_j(q)d\hat{H}_j^{Worst}(q).\end{align*}

For the general case that $x_j(h)$ $C$-dominates $\hat{x}_j(h)$, we construct $\bar{x}_j(h)=\begin{cases}\hat{x}_j(h/C),&h\in [0,C],\\0,&h\in(C,1]\end{cases}$, let $\bar{H}_j^{Worst}$ and $\bar{H}_j^{Best}$ be the worst and best cumulative equilibrium in $\mathcal H_j^*(\bar{x}_j(h),x_{-j}(h))$. We can apply the second case on $\bar{x}_j(h)$ and $\hat{x}_j(h)$, then apply the first case on $x_j(h)$ and $\bar{x}_j(h)$, then we get \begin{align*}\int_0^1v_j(q)dH_j^{Best}(q)&\geq\int_0^1v_j(q)d\bar{H}_j^{Best}(q)\geq C\int_0^1v_j(q)d\hat{H}_j^{Best}(q),\\ \int_0^1v_j(q)dH_j^{Worst}(q)&\geq\int_0^1v_j(q)d\bar{H}_j^{Worst}(q)\geq C\int_0^1v_j(q)d\hat{H}_j^{Worst}(q).\end{align*}

By statement 1, to show statement 2, we only need to prove the case of $C=1$.
Suppose $\forall h\in[0,1]$, $x_j(h)>\hat{x}_j(h)$.
Let $Q(x)=x_{j}^{-1}(x)+x_{-j}^{-1}(x),\hat{Q}(x)=\hat{x}_{j}^{-1}(x)+x_{-j}^{-1}(x)$. For any $X\in[\hat{Q}^{-1}(1),\hat{x}_j(0)]$, $x_j^{-1}(X)\geq \hat{x}_j^{-1}(X)$, and $Q(X)\geq \hat{Q}(X)$.

For any $q\in[0,1]$, let $X=Q^{-1}(q)$, $\hat{X}=\hat{Q}^{-1}(q)$, then $X\geq\hat{X}$. If $X>\hat{X}$, then $q-\hat{H}_j^{Best}(q)\geq x_{-j}^{-1}(\hat{X}+0)\geq x_{-j}^{-1}(X)\geq q-H_j^{Worst}(q)$. 
If $X=\hat{X}$, then $\hat{H}_j^{Best}(q)\leq \hat{x}_j^{-1}(\hat{X})=\sup\{h\in[0,1]:\hat{x}_j(j)\geq \hat{X}\}\leq \sup\{h\in[0,1]:x_j(j)>\hat{X}\}=x_j^{-1}(X+0)\leq H_j^{Worst}(q)$.
Thus $\hat{H}_j^{Best}(q)\leq H_j^{Worst}(q)$, so by \Cref{lemma:integral-compare}, $\int_0^1v_j(q)dH_j^{Worst}(q)\geq C\int_0^1v_j(q)d\hat{H}_j^{Best}(q)$.
\qed
\end{proof}



\subsection{Proof of \Cref{proposition:best-response-trivial-lowerbound}}
\begin{proof}
When $x_j(h)$ $C$-strongly-dominates $x_{-j}(h)$, consider any cumulative equilibrium $(H_j(q),H_{-j}(q))\in\mathcal H^*(x_j(h),x_{-j}(h))$. Since $\int_0^1v_j(q)dq =\int_0^1v_j(q)dH_j(q)+\int_0^1v_j(q)dH_{-j}(q)$, we only need to prove that $\int_0^1v_j(q)dH_j(q)\geq C\int_0^1v_j(q)dH_{-j}(q)$.

Let $\underline{X}=Q^{-1}(1)=\inf\{X:x_j^{-1}(X)+x_{-j}^{-1}(X)\leq 1\}$. For any $X\in[\underline{X},x_{-j}(0)]$, whenever $x_{-j}(h)\geq X$, $x_j(Ch)>X$. So we obtain $x_j^{-1}(X+0)=\sup\{h\in[0,1]:x_{j}(h)>X\}\geq\sup\{\min(1,Ch):h\in[0,1],x_{-j}(h)\geq X\}=\min(1,C\sup\{h\in[0,1]:x_{-j}(h)\geq X\})=\min(1,Cx_{-j}^{-1}(X))$.

For any $q\in[0,1]$, by \Cref{theorem:player-equilibrium-general-condition}, $x_j^{-1}(Q^{-1}(q)+0)\leq H_j(q)\leq x_j^{-1}(Q^{-1}(q))$. If $H_j(q)<1$, then $H_j(q)\geq x_j^{-1}(Q^{-1}(q)+0)\geq Cx_{-j}(Q^{-1}(q))\geq CH_{-j}(q)$. If $H_j(q)=1$, we still have $H_j(q)>0=CH_{-j}(q)$. Therefore, $H_j(q)\geq CH_{-j}(q)$ and $\int_0^1v_j(q)dH_j(q)\geq \int_0^1v_j(q)d(CH_{-j}(q))=C\int_0^1v_j(q)dH_{-j}(q)$.

When $x_j(h)$ $C$-dominates $x_{-j}(h)$, take $H_j^{Best}(q)$ as in \Cref{proposition:best-response-worst-and-best-equilibrium} and let $H_{-j}(q)=q-H_j^{Best}(q)$. We only need to prove that $\int_0^1v_j(q)dH_j^{Best}(q)\geq C\int_0^1v_j(q)$ $dH_{-j}(q)$ holds.

Let $\underline{X}=Q^{-1}(1)=\inf\{X:x_j^{-1}(X)+x_{-j}^{-1}(X)\leq 1\}$. For any $X\in[\underline{X},x_{-j}(0)]$, whenever $x_{-j}(h)\geq X$, $x_j(Ch)\geq X$, so we obtain $x_j^{-1}(x)=\sup\{h\in[0,1]:x_{j}(h)\geq x\}\geq \min(1,C\sup\{h\in[0,1]:x_{-j}(h)\geq x\})=\min(1,Cx_{-j}^{-1}(x))$, and $x_j^{-1}(x+0)\geq\min(1,Cx_{-j}^{-1}(x+0))$.

For every $q\in[0,1]$, let $X=Q^{-1}(q)$. By \Cref{proposition:best-response-worst-and-best-equilibrium}, if $q\leq x_j^{-1}(X)+x_{-j}^{-1}(X+0)$, we have $H_j^{Best}(q)=q-x_{-j}^{-1}(X+0)\geq x_j^{-1}(X)$ and $H_{-j}(q)=x_{-j}^{-1}(X+0)$, so $H_j^{Best}(q)\geq CH_{-j}(q)$. If $q>x_j^{-1}(X)+x_{-j}^{-1}(X+0)$, then $H_j^{Best}(q)=x_j^{-1}(X)$ and $H_{-j}(q)\leq x_{-j}^{-1}(X)$, so $H_j^{Best}(q)\geq CH_{-j}(q)$.
In either case, $H_j^{Best}(q)\geq CH_{-j}(q)$, so $\int_0^1v_j(q)dH_j(q)\geq \int_0^1v_j(q)d(CH_{-j}(q))=C\int_0^1v_j(q)dH_{-j}(q)$.

Finally, let $x_j(h)=\begin{cases}
x_{-j}(\frac{t_{-j}}{t_{j}}h),&\text{if }h\in[0,\min\{\frac{t_j}{t_{-j}},1\}]\\
0,&\text{if }h>\min\{\frac{t_j}{t_{-j}},0\}\end{cases}$. Since $x_j(h)$ ($\frac{t_j}{t_{-j}}$)-dominates $x_{-j}(h)$, we have $$R^{Best}(x_j(h))\geq\frac{t_j}{t_j+t_{-j}}\int_0^1v_j(q)dq.$$
\qed
\end{proof}



\section{Missing Proofs in Section \ref{section:best_response}}
\label{app:sec5}

\subsection{Proof of NP-Hardness (Theorem \ref{theorem:NP-hard})}

\begin{proof}
Consider the following special case: assume that both $v_j(q)$ and $x_{-j}(h)$ are piecewise constant, i.e., there is some intervals forming a partition of $[0,1]$, and the function value is a constant on each interval. Without loss of generality we assume that both $v_j(q)$ and $x_{-j}(h)$ are right-continuous. It can be seen easily that they can be represented precisely by the lengths of and function values on the intervals.

Let $a_1,a_2,\cdots,a_N$ each denotes the length of an interval on which $x_{-j}(h)$ is a positive constant, where $a_i>0$, and $\sum_{i=1}^N a_i\leq 1$. Assume $c_1> c_2> \cdots > c_N> 0$ are the function values of $x_{-j}(h)$ on these intervals. Let $A_i=\sum_{j=1}^i a_j$, and $A_0=0$. We formally define $x_{-j}(h)=\begin{cases}c_i,&\forall h\in [A_{i-1},A_i),i=1,2,\cdots,N\\
0,&\forall h\in[A_N,1]\end{cases}$.


Given $t_j$, contest designer $j$ tries to find a best response $x_j(h)\in\mathcal F_{t_j}$ to maximize $R^{Best}(x_j(h)|x_{-j}(h),v_j(q))$.

We first observe that, without loss of generality, we may assume that there is $0=B_0\leq B_1\leq B_2\leq\cdots\leq B_N\leq 1$, such that \begin{equation}\label{eq:assumption-Bi-proofofNPhard}
x_j(h)=\begin{cases}c_i,&\forall h\in [B_{i-1},B_i),i=1,2,\cdots,N\\
0,&\forall h\in[B_N,1]\end{cases}
\end{equation}

Given any $\hat{x}_j(h)$, let $\bar{x}_j(h)=\max\{x\in\{c_1,c_2,\cdots,c_N,0\}:x\leq\hat{x}_j(h)\}$, and let $x_j(h)=\bar{x}_j(h+0)$, then $x_j(h)$ satisfies the assumption (\ref{eq:assumption-Bi-proofofNPhard}). We show that $R^{Best}(x_j(h)|x_{-j}(h),v_j(q))=R^{Best}(\hat{x}_j(h)|x_{-j}(h),v_j(q))$.

For any $i=1,2,\cdots,N$, we have $x_j^{-1}(c_i)=\bar{x}^{-1}(c_i)=\hat{x}^{-1}(c_i)$. Since for any $(H_j(q),H_{-j}(q))\in\mathcal H^*(x_j(q),x_{-j}(q))$, it holds that $\int_0^1 v_j(q)dq=\int_0^1 v_j(q)dH_j(q)+\int_0^1 v_j(q)dH_{-j}(q)$, we have $$R^{Best}(x_j(h)|x_{-j}(h),v_j(q))+R^{Worst}(x_{-j}(h)|x_j(h),v_j(q))=\int_0^1v_j(q)dq.$$ From \Cref{proposition:best-response-worst-and-best-equilibrium}, we know that the worst equilibrium for the opponent $-j$ is $H_{-j}^{Worst}(q)=\max\{x_{-j}^{-1}(Q^{-1}(q)+0),q-x_{j}^{-1}(Q^{-1}(q))\}$. By some calculation we have 
$$H_{-j}^{Worst}(q)=\begin{cases}
A_{i-1},&\forall q\in [x_j^{-1}(c_{i-1})+A_{i-1},x_j^{-1}(c_i)+A_{i-1}],\\
q-x_j^{-1}(c_i),&\forall q\in [x_j^{-1}(c_i)+A_{i-1},x_j^{-1}(c_i)+A_i]
\end{cases}$$ where we assume $c_0=+\infty$.
Observe that $H_{-j}^{Worst}(q)$ only depends on the value of $x_j^{-1}(c_i)$, where $i=1,2,\cdots$. Recall that $x_j^{-1}(c_i)=\hat{x}_j^{-1}(c_i)$, we have $$R^{Worst}(x_{-j}(h)|x_j(h),v_j(q))=R^{Worst}(x_{-j}(h)|\hat{x}_j(h),v_j(q)),$$and therefore $$R^{Best}(x_j(h)|x_{-j}(h),v_j(q))=R^{Best}(\hat{x}_j(h)|x_{-j}(h),v_j(q)).$$
Note that $\forall h\in[0,1], x_j(h)\leq\bar{x}_j(h)\leq \hat{x}_j(h)$, therefore $\int_0^1x_j(h)dh\leq\int_0^1\hat{x}_j(h)dh$, so it is not worse to replace $\hat{x}_j(h)$ with $x_j(h)$, which satisfies the assumption.

Let $B_{N+1}=1$ and $c_{N+1}=0$. Under assumption (\ref{eq:assumption-Bi-proofofNPhard}), for any $x\in[0,c_1]$, $x_{-j}^{-1}(x)=\max\{A_i:i\in\{1,2,\cdots,N\},c_i\geq x\}$, and $x_{j}^{-1}(x)=\max\{B_i:i\in\{1,2,\cdots,N\},c_i\geq x\}$.
For convenience define $v_j(q)=0$ for any $q>1$. Then \begin{align*}R^{Best}(x_j(h)|x_{-j}(h),v_j(q))&=\sum_{i=1}^{N+1}\int_{A_{i-1}+B_{i-1}}^{A_{i-1}+B_i}v_j(q)dq\\
&=\int_0^1v_j(q)dq-\sum_{i=1}^{N}\int_{B_i+A_{i-1}}^{B_i+A_{i}}v_j(q)dq.\end{align*}
Note that the budget constraint $\int_0^1x_j(h)dh\leq t_j$ can be written as $$t_j\geq\sum_{i=1}^N c_i(B_{i}-B_{i-1})=\sum_{i=1}^N (c_i-c_{i+1})B_{i}.$$ Let $d_i$ denote $c_i-c_{i+1}$.
Therefore, finding the optimal $x_j(h)$ is equivalent to finding the minimum solution of the following optimization problem:
\begin{align}
\min_{B_1,B_2,\cdots,B_N}&\sum_{i=1}^{N}\int_{B_i+A_{i-1}}^{B_i+A_{i}}v_j(q)dq,\label{eqn:NP-hard-optimize-problem} \nonumber \\
\text{s.t.} ~~~~~ &\sum_{i=1}^N d_iB_i\leq t_j \\
&0\leq B_1\leq B_2\leq\cdots\leq B_N\leq 1 \nonumber
\end{align}

Now we prove that the optimization problem (\ref{eqn:NP-hard-optimize-problem}) is NP-hard by reducing knapsack problem to it.

Consider the knapsack problem: Input positive integers $N',W,U$, representing the number of items, the capacity and the target total value. Then set $2N'$ positive integers $w_i,u_i$, where each pair represents the weight and value of the $i$-th item, respectively. The problem is to decide whether there exists $z_1,z_2,\cdots,z_{N'}$, such that $z_i\in\{0,1\}$ for all $i$, $\sum_{i=1}^{N'}z_iw_i\leq W$ and $\sum_{i=1}^{N'}z_iu_i\geq U$.

Given an instance of the knapsack problem, without loss of generality, we assume $1\leq w_i\leq W$ and $1\leq u_i\leq U$ holds for all $i$.
We construct an instance of problem (\ref{eqn:NP-hard-optimize-problem}).

We set $N=N'$. Let $a=\frac1{2WN(3N)^N}$, and set $a_1=a_2=\cdots=a_N=a$, consequently $A_i=ia$.

Set $t_j=2W$.

For $i=1,2,\cdots,N$, let $L_i=\frac{(3N)^i}{2N(3N)^N}$, set $d_i=\frac{W}{NL_i}$. Let $l_i=\frac{w_i}{d_i}$.

Note that $l_i\geq \frac{1}{d_i}\geq \frac{NL_i}{W}\geq \frac{3N}{2W(3N)^N}>a$, and for $i=1,2,\cdots,N-1$, $L_i+l_i=L_i+\frac{w_i}{W}NL_i\leq (N+1)L_i<3NL_i=L_{i+1}$.

Roughly speaking, we will construct a $v_j(q)$ such that, when $B_i$ move from $L_i-a$ to $L_i$, $\int_{B_i+A_{i-1}}^{B_i+A_{i}}v_j(q)dq$ greatly decreases, which forces $B_i$ to be at least $L_i$; and when $B_i$ moves from $L_i+l_i-a$ to $L_i+l_i$, $\int_{B_i+A_{i-1}}^{B_i+A_{i}}v_j(q)dq$ is reduced by an amount proportional to $u_i$, the value of the $i$-th item, indicating that the $i$-th item is selected. Note that the cost of moving $B_i$ from $L_i$ to $L_i+l_i$ is $l_id_i=w_i$, exactly equals to the weight of the $i$-th item.

Let $M=\max\{(3N)^NU,\sum_{i=1}^Nu_i\}+1$. Construct $v_j(q)$ as the following piecewise constant function:\\
$$v_j(q)=\begin{cases}
NM,&\forall q\in[0,L_1);\\
(N-i)M+u_i,&\forall q\in \left[L_i+(i-1)a,L_i+l_i+(i-1)a\right),1\leq i\leq N;\\
(N-i)M,&\forall q\in\left[L_i+l_i+(i-1)a,L_{i+1}+ia\right),1\leq i\leq N-1;\\
0,&\forall q\in[L_N+l_N+(N-1)a,1].
\end{cases}$$

Now we prove that, the instance of knapsack problem has a solution that $\sum_{i=1}^{N'}z_iu_i\geq U$, if and only if the optimal value of problem (\ref{eqn:NP-hard-optimize-problem}) satisfies $$\min_{B_1,B_2,\cdots,B_N}\sum_{i=1}^{N}\int_{B_i+A_{i-1}}^{B_i+A_{i}}v_j(q)dq\leq \left(\frac{N(N-1)}{2}M+\sum_{i=1}^Nu_i-U\right)a.$$

If there exists $z_1,z_2,\cdots,z_{N}$, such that $z_i\in\{0,1\}$ for all $i$, and that $\sum_{i=1}^{N}z_iw_i$ $\leq W$, and that $\sum_{i=1}^{N}z_iu_i\geq U$, let $B_i=L_i+z_il_i$.

For any $i=1,2,\cdots,N$, if $z_i=0$, then $B_i+A_{i-1}=L_i+(i-1)a$, and $B_i+A_{i}=L_i+ia<L_i+l_i+(i-1)a$, so $v_j(q)=(N-i)M+u_i$ for any $q\in[B_i+A_{i-1},B_i+A_{i}]$, then $\int_{B_i+A_{i-1}}^{B_i+A_{i}}v_j(q)dq=((N-i)M+u_i)a$. If $z_i=1$, then $B_i+A_{i-1}=L_i+l_i+(i-1)a$, and $B_i+A_{i}=L_i+l_i+ia<L_{i+1}$, so $v_j(q)=(N-i)M$ for any $q\in[B_i+A_{i-1},B_i+A_{i}]$, then $\int_{B_i+A_{i-1}}^{B_i+A_{i}}v_j(q)dq=((N-i)M)a$. Thus we have $\sum_{i=1}^{N}\int_{B_i+A_{i-1}}^{B_i+A_{i}}v_j(q)dq=(\frac{N(N-1)}{2}M+\sum_{i=1}^N(1-z_i)u_i)a\leq(\frac{N(N-1)}{2}M+\sum_{i=1}^Nu_i-U)a$.

Since $\sum_{i=1}^Nd_iB_i=\sum_{i=1}^N(d_iL_i+z_id_il_i)=\sum_{i=1}^N(\frac{W}{N}+z_iw_i)=W+\sum_{i=1}^Nz_iw_i\leq W+W=t_j$, and $B_i<L_{i+1}\leq B_{i+1}$, and $B_N\leq L_N+l_N<1$, $B_1,B_2,\cdots,B_N$ is a feasible solution of problem (\ref{eqn:NP-hard-optimize-problem}), therefore $\min_{B_1,B_2,\cdots,B_N}\sum_{i=1}^{N}\int_{B_i+A_{i-1}}^{B_i+A_{i}}v_j(q)dq\leq (\frac{N(N-1)}{2}M+\sum_{i=1}^Nu_i-U)a$.

If $\min_{B_1,B_2,\cdots,B_N}\sum_{i=1}^{N}\int_{B_i+A_{i-1}}^{B_i+A_{i}}v_j(q)dq\leq (\frac{N(N-1)}{2}M+\sum_{i=1}^Nu_i-U)a$, we prove that a solution  $z_1,z_2,\cdots,z_{N}$ of the knapsack problem exists.

By continuity of the object function $\sum_{i=1}^{N}\int_{B_i+A_{i-1}}^{B_i+A_{i}}v_j(q)dq$ and compactness of the feasible solution space of problem (\ref{eqn:NP-hard-optimize-problem}), there must exist a feasible minimum solution $B_1,B_2,\cdots,B_N$ such that $\sum_{i=1}^{N}\int_{B_i+A_{i-1}}^{B_i+A_{i}}v_j(q)dq\leq (\frac{N(N-1)}{2}M+\sum_{i=1}^Nu_i-U)a$.

Suppose (for a contradiction) that there is $i'\in\{1,2,\cdots,N\}$ such that $B_{i'}\leq L_{i'}-a$, then $B_{i'}+A_{i'}\leq L_{i'}+(i-1)a_{i'}$, so $\int_{B_i+A_{i-1}}^{B_i+A_{i}}v_j(q)dq\geq (N-i'+1)Ma$. Since for any $i\in\{1,2,\cdots,N\}$, $B_{i}\leq \frac{t_j}{d_i}=\frac{2WNL_i}{W}=NL_i<L_{i+1}-a$, we have $\sum_{i=1}^N\int_{B_i+A_{i-1}}^{B_i+A_{i}}v_j(q)dq\geq (M+\sum_{i=1}^N(N-i)M)a=(\frac{N(N-1)}{2}M+M)a>(\frac{N(N-1)}{2}M+\sum_{i=1}^Nu_i)a$, which contradicts with $\sum_{i=1}^{N}\int_{B_i+A_{i-1}}^{B_i+A_{i}}v_j(q)dq\leq (\frac{N(N-1)}{2}M+\sum_{i=1}^Nu_i-U)a$. Therefore, for any $i\in\{1,2,\cdots,N\}$, it must hold that $B_i>L_i-a$.

Next, suppose that there is $i\in\{1,2,\cdots,N\}$ such that $L_{i}-a<B_i<L_i$. $\int_{B_i+A_{i-1}}^{B_i+A_{i}}v_j(q)dq=(L_i-B_i)(N-i+1)M+(a-L_i+B_i)((N-i)M+u_i)=((N-i)M+u_i)a+(M-u_i)(L_i-B_i)$.
If $\sum_{i'=1}^Nd_{i'}B_{i'}<t_j$, there is sufficiently small $\epsilon>0$ such that $\sum_{i'=1}^Nd_{i'}B_{i'}+\epsilon\leq t_j$, and that $B_i<B_i+\frac{\epsilon}{d_i}\leq L_i$. Let $B_i^*=B_i+\frac{\epsilon}{d_i}$, then $\int_{B_i^*+A_{i-1}}^{B_i^*+A_{i}}v_j(q)dq=((N-i)M+u_i)a+(M-u_i)(L_i-B_i-\frac{\epsilon}{d_i})<\int_{B_i+A_{i-1}}^{B_i+A_{i}}v_j(q)dq$. Therefore by replacing $B_i$ with $B_i^*$ we get a strictly better solution, which is a contradiction. 

Now assume that $\sum_{i'=1}^Nd_{i'}B_{i'}=t_j$, since $\sum_{i'=1}^Nd_{i'}L_{i'}=W<t_j$, there must exist $i''$ such that $B_{i''}>L_{i''}$. Let $f_1(\epsilon)=\int_{B_{i''}+A_{{i''}-1}-\frac{\epsilon}{d_{i''}}}^{B_{i''}+A_{i''}-\frac{\epsilon}{d_{i''}}}v_j(q)dq$, $f_2(\epsilon)=\int_{B_i+A_{i-1}+\frac{\epsilon}{d_i}}^{B_i+A_{i}+\frac{\epsilon}{d_i}}v_j(q)dq$. Then we have the right-hand derivatives
\begin{align*}
{f_1}'_+(0) & =\frac{-v_j(B_{i''}+A_{{i''}}-0)+v_j(B_{i''}+A_{{i''}-1}-0)}{d_{i''}}\\
& \leq \frac{-(N-i'')M+((N-i'')M+u_{i''})}{d_{i''}}=\frac{u_{i''}}{d_{i''}}, \text{and} \\ {f_2}'_+(0)& =\frac{v_j(B_i+ia+0)-v_j(B_i+(i-1)a+0)}{d_{i}}\\
&=\frac{((N-i)M+u_i)-(N-i+1)M}{d_i}=-\frac{M-u_i}{d_i}. 
\end{align*}
Since $\frac{M-u_i}{d_i}=\frac{(M-u_i)W}{NL_i}$, $\frac{u_{i''}}{d_{i''}}=\frac{u_{i''}W}{NL_{i''}}$, we have $\frac{M-u_i}{d_i}\frac{d_{i''}}{u_{i''}}=\frac{(M-u_i)L_{i''}}{u_{i''}L_i}=\frac{M-u_i}{u_{i''}}(3N)^{i''-i}\geq \frac{(3N)^NU-u_i}{U}(3N)^{1-N}\geq 3N-1>1$, so $\frac{M-u_i}{d_i}>\frac{u_{i''}}{d_{i''}}$, i.e. ${f_1}'_+(0)+{f_2}'_+(0)>0$. Then there is sufficiently small $\epsilon>0$ such that $f_1(\epsilon)+f_2(\epsilon)<0$. That is, if we replace $B_{i}$ with $B_{i}+\frac{\epsilon}{d_i}$, and replace $B_{i''}$ with $B_{i''}-\frac{\epsilon}{d_{i''}}$, we get a strictly better solution, which contradicts.

Therefore, for all $i\in\{1,2,\cdots,N\}$, we have $B_i\in [L_i,L_{i+1}-a)$. Let $\hat{z}_i=\begin{cases}0,&\text{if }B_i<L_i+l_i-a;\\1,&\text{if }B_i\geq L_i+l_i-a\end{cases}$, then \begin{align*}
    \sum_{i=1}^N\hat{z}_iu_ia
\geq\sum_{i=1}^{N}((N-i)Ma+u_ia-\int_{B_i+A_{i-1}}^{B_i+A_{i}}v_j(q)dq)\\
\geq (\frac{N(N-1)}{2}M+\sum_{i=1}^Nu_i)a-\sum_{i=1}^{N}\int_{B_i+A_{i-1}}^{B_i+A_{i}}v_j(q)dq
\geq Ua.
\end{align*}
Note that for any $i\in\{1,2,\cdots,N\}$, $\hat{z}_iw_i\leq \hat{z}_il_id_i\leq (B_i-L_i+a)d_i$. Therefore, we have
\begin{align*}
\sum_{i=1}^N\hat{z}_i &\leq\sum_{i=1}^N(B_i-L_i+a)d_i\leq t_j-W+\sum_{i=1}^Nad_i \\ 
& \leq W +\sum_{i=1}^N\frac{1}{2WN(3N)^N}\frac{2N(3N)^NW}{N(3N)}\leq W+\frac1{3N}<W+1. 
\end{align*}Since $\hat{z}_i\in\{0,1\}$, we get that $\sum_{i=1}^N\hat{z}_i\leq W$, and $\sum_{i=1}^N\hat{z}_iu_i\geq U$, which satisfies the knapsack problem.

One can easily see that the reduction can be done in polynomial time.
\qed
\end{proof}

\newpage
\subsection{The Algorithm in FPTAS}
\begin{algorithm}[htbp!]
\caption{FPTAS for the best-response problem}\label{alg:FPTAS}
\begin{algorithmic}[1]
\State \textbf{Input} $\epsilon>0$
\Procedure{DP}{$K^*$}
\State $\lambda_v\gets\frac{\epsilon'K^*}p$
; $\rho_v\gets\lceil \frac{2K^*}{\lambda v}\rceil$
\State $\lambda_t\gets\frac{\epsilon't_j}{p}$
; $\rho_t\gets\lfloor\frac{(1+\epsilon')^3t_j}{\lambda_t}\rfloor$.
\State \textbf{Initialize} $b[0\dots p,0\dots \rho_v,0\dots\rho_t]$ as $b[i,r,s]=+\infty\forall i,r,s$
\State $b[0,0,0]\gets 0$

\For{$i\gets 1,\cdots,p$}
\For{$r\gets 0,\cdots,\rho_v$}
\For{$s\gets 0,\cdots,\rho_t$}
    \State $B^*\gets b[i-1,r,s]$
    \State $A^*\gets x_{-j}^{-1}(c_i+0)$
    \If{$B^*\leq 1$}
        \For{$r'\gets r,\cdots,\rho_v$}
            \If{$(r'-r)\lambda_v>K-V(B^*+A^*)$}
                \State\textbf{break}
            \EndIf
            \State $\delta\gets V^{-1}(V(B^*+A^*)+(r'-r)\lambda_v)-B^*-A^*$
            \State $s'\gets s+\lceil\frac{\delta c_i}{\lambda_t}\rceil$
            \If{$s'\leq \rho_t\land b[i,r',s']>B^*+\delta$}
                \State $b[i,r',s']\gets B^*+\delta$
                \State $\mathrm{hist}[i,r',s']\gets (r,s)$\Comment{auxiliary infomation}
            \EndIf
        \EndFor
    \EndIf
\EndFor
\EndFor
\EndFor

\For{$r\gets \rho_v,\rho_v-1,\cdots,0$}
\Comment{Recover the best solution found}
    \State $flag\gets \mathrm{false}$
    \For{$s\gets \rho_t,\rho_t-1,\cdots,0$}
        \If{$b[p,r,s]\leq 1$}
            \State $flag\gets \mathrm{true}$
            \State $s^*\gets s$
            \State \textbf{break}
        \EndIf
    \EndFor
    \If{$flag==\mathrm{true}$}
        \State $r^*\gets r$
        \State\textbf{break}
    \EndIf
\EndFor

\algstore{bkbreak}
\end{algorithmic}
\end{algorithm}
\begin{algorithm}
\begin{algorithmic}[1]
\algrestore{bkbreak}

\State $R^*\gets \lambda_v*r^*$
\For{$i\gets p,p-2,\cdots,1$}
    \State $B_i\gets b[i,r^*,s^*]$
    \State $r^*,s^*\gets \mathrm{hist}[i,r^*,s^*]$
\EndFor
\State $\vec{B}^*\gets(B_1,B_2,\cdots,B_p)$
\State \textbf{return} $\vec{B}^*,R^*$
\EndProcedure

\Comment{Main body, try possible $K^*$.}
\State $\vec{B}^*,R\gets \mathrm{NULL},0$
\State $\mu\gets\frac{t_j}{t_j+t_{-j}}$
\While{$\mu\leq1$}
    \State $\vec{B}',R'\gets$\Call{DP}{$\mu*K$}
    \If{$R'>R$}
    \State $\vec{B}^*,R\gets\vec{B}',R'$
    \EndIf
    \State $\mu\gets\mu*2$
\EndWhile

\For{$i\gets1,\cdots,p$}
\State $B_i\gets(1-\epsilon')^3B_i^*$\Comment{Suppose $\vec{B}^*=(B_1^*,B_2^*,\cdots,B_p^*)$}
\EndFor
\State $\vec{B}\gets(B_1,B_2,\cdots,B_p)$
\State \textbf{output} $\vec{B}$
\end{algorithmic}
\end{algorithm}

\newpage
\subsection{Proof of \Cref{theorem:FPTAS}}
\begin{proof}
Given $v_j(q)$ and $x_{-j}(h)$ and $t_j$, let $M=v_j(0)$, $K=\int_0^1v_j(q)dq$, $L=x_{-j}(0)$, $t_{-j}=\int_0^1x_{-j}(h)dh$.

We first describe a dynamic programming (DP) procedure, which try to find an approximately best result with the hint of an input parameter $K^*$.

Given any $\epsilon>0$, let $\epsilon'=\frac14\epsilon$. We use dynamic programming to find a good interim allocation function in a restricted solution space, denoted by $S$.

We use $S$ to denote the solution space, which consists of all interim allocation functions $x_j(h)$ satisfies the following restrictions:
\begin{enumerate}
    \item $x_j(h)\in\mathcal F_{(1+\epsilon')^3t_j}$.
    \item Let $\lambda_x=\epsilon't_j$, and $p=\lfloor\log_{(1+\epsilon')}(\frac{L}{\lambda_x})\rfloor+1$. 
For $i=1,2,\cdots,p$, let $c_i=L(1+\epsilon')^{-i+1}$.\\
For any $h\in[0,1]$, we restrict that $x_j(h)\in\{c_1,c_2,\cdots,c_p,0\}$.\\
In other words, let $B_i$ denote $x_j^{-1}(c_i)$ for $i=1,2,\cdots,p$, and assume $B_0=0$, $B_{p+1}=1$, we have $\forall i=1,2,\cdots,p$, $\forall h\in(B_{i-1},B_i)$, $x_j(h)=c_i$. We call $(B_{i-1},B_i)$ the $i$-th segment.
    \item For convenience we define $v_j(q)=0$ for any $q>1$, then we have\\ $$R^{Best}(x_j(h)|x_{-j}(h),v_j(q))=\sum_{i=1}^p\int_{x_{-j}^{-1}(c_i+0)+B_{i-1}}^{x_{-j}^{-1}(c_i+0)+B_i}v_j(q)dq.$$
    We restrict that for any $i=1,2,\cdots,p$, $\frac1{\lambda_v}\int_{x_{-j}^{-1}(c_i+0)+B_{i-1}}^{x_{-j}^{-1}(c_i+0)+B_i}v_j(q)dq$ is a non-negative integer, where $\lambda_v:=\frac{\epsilon'K^*}p$.
    \item We call $(B_{i}-B_{i-1})c_i$ the budget requirement of the $i$-th segment. Let $\lambda_t:=\frac{\epsilon't_j}{p}$, we restrict that \begin{equation}\label{eq:budget-round-up-FPTAS}
\sum_{i=1}^p{\lambda_t}\left\lceil\frac{(B_{i}-B_{i-1})c_i}{\lambda_t}\right\rceil\leq (1+\epsilon')^3t_j.
    \end{equation}Intuitively, if we round up the budget requirement of every segment to the closest multiple of $\lambda_t$, the total budget will not exceed $(1+\epsilon')^3t_j$.
\end{enumerate}

Intuitively, $\lambda_v$ is the minimum unit of the utility obtained by each segment, $\lambda_t$ is the minimum unit of the budget requirement (after rounding up) of each segment.

We define $\rho_v=\lceil \frac{2K^*}{\lambda_v}\rceil$, $\rho_t=\lceil\frac{(1+\epsilon')^3t_j}{\lambda_t}\rceil$. For any $x_j(h)$ in $S$, if $R^{Best}(x_j(h))\leq 2K^*$, then there is some integer $r\in\{0,1,\cdots,\rho_v\}$ such that $R^{Best}(x_j(h))=r\lambda_v$. And by constraint (\ref{eq:budget-round-up-FPTAS}), for any $x_j(h)$ in $S$, there is some integer $s\in\{0,1,\cdots,\rho_t\}$ such that $\sum_{i=1}^p\lceil\frac{(B_{i}-B_{i-1})c_i}{\lambda_t}\rceil=s$.

For $i=0,\cdots,p$, $r=0,\cdots,\rho_v$, $s=0,\cdots,\rho_t$, for the state $(i,r,s)$, we define the subproblem $b(i,r,s)$ as \begin{align*}
    b(i,r,s)=&\min_{0\leq B_1\leq B_2\leq \cdots \leq B_i\leq 1}B_i,\\
    \text{s.t.}&\sum_{i'=1}^i\int_{x_{-j}^{-1}(c_{i'}+0)+B_{{i'}-1}}^{x_{-j}^{-1}(c_{i'}+0)+B_{i'}}v_j(q)dq=r\lambda_v,\\
    &\sum_{i'=1}^i\lceil\frac{(B_{i'}-B_{i'-1})c_{i'}}{\lambda_t}\rceil=s.
\end{align*}
When there is no feasible solution, we define $b(i,r,s)=+\infty$. Intuitively, $b(i,r,s)$ denotes the minimum $B_i$, such that we have found $B_1,\cdots,B_i$, representing the first $i$ segments of some $x_j(h)$ in $S$, such that the total utility obtained by these $i$ segments is $r\lambda_v$, and that the total budget requirement (after rounding up for each segment) of these $i$ segments is $s\lambda_t$. The way to calculate the table of all $b(i,r,s)$ is shown in \Cref{alg:FPTAS}. After calculating the table, to find the $x_j(h)$ in $S$ with the maximum $R^{Best}(x_j(h))$, we only need to find the biggest $r\in\{0,\cdots,\rho_v\}$ such that $\exists s\in\{0,\cdots,\rho_t\}$, $b(p,r,s)\leq 1$. Denote this biggest $r$ by $r^*$, we can reconstruct the $B_1,\cdots,B_p$ corresponding with $x_j(h)$ such that $R^{Best}(x_j(h))=r^*\lambda_v$.

The DP procedure finally output the best $B_1,\cdots,B_p$ found, together with its utility $r^*\lambda_v$.

The main body of \Cref{alg:FPTAS} calls the DP procedure for $O(\log_2\frac{t_j+t_{-j}}{t_j})$ times. The parameter $K^*$ starts from $\frac{t_j}{t_j+t_{-j}}K$, and is doubled each time, until it exceeds $K$.
Intuitively, since the utility of the best response $x_j(h)\in\mathcal F_{t_j}$, denoted by $R^*$, satifies that $\frac{t_j}{t_j+t_{-j}}K\leq R^*\leq K$, there is a $K^*$ with which the DP procedure is called for one time, such that $K^*\leq R^*< 2K^*$. We will explain this later in the analysis of correctness.

Among the results $((B_1,\cdots,B_p),r^*\lambda_v)$ returned by the DP procedure, the one which has the highest utility is recorded. Since the recorded $(B_1,\cdots,B_p)$ corresponds to a $x_j(h)$ in $S\subset \mathcal F_{(1+\epsilon')^3t_j}$, we use a horizontal stretch method to fit in the budget constraint $t_j$. After scaling each $B_i$ to $(1-\epsilon')^{3}B_i$, we get a $x_j(h)\in\mathcal F_{t_j}$, which is the final output of our algorithm.

Now we prove that \Cref{alg:FPTAS} is a FPTAS.

The total running time of \Cref{alg:FPTAS} is mainly composed of $\lfloor\log_2(1/\frac{t_j}{t_j+t_{-j}})\rfloor+1$ times calling the DP procedure. The total number of states in the dynamic programming is $O(p\rho_v\rho_t)$, and each call to the DP procedure runs in $O(p\rho_v^2\rho_t)$ time. Since $\lfloor\log_2(1/\frac{t_j}{t_j+t_{-j}})\rfloor+1=O(\log(1+\frac{t_{-j}}{t_j}))$, $p=O(\frac1{\epsilon'}\ln(\frac{L}{t_j\epsilon'}))$, $\rho_v=O(\frac{p}{\epsilon'})$, $\rho_t=O(\frac{p}{\epsilon'})$, we obtain that the time complexity of \Cref{alg:FPTAS} is $O(\log(1+\frac{t_{-j}}{t_j})\epsilon^{-7}\ln^4(\frac{L}{t_j\epsilon}))$, which is polynomial in the input size and $\frac1{\epsilon}$. The space complexity is $O(p\rho_v\rho_t)$, which is also polynomial in the input size and $\frac1{\epsilon}$.

Now we prove the $(1-\epsilon)$ approximation ratio. 



Suppose the best response solution is $x^*_j(h)\in\mathcal F_{t_j}$, and let $R^*$ denote $R^{Best}(x^*_j(h))$\footnote{For convenience we assume that $\sup_{x_j(h)\in\mathcal F_{t_j}}R^{Best}(x_j(h))$ is attained by $x_j^*(h)$. If the supremum is not attained by any $x_j(h)\in\mathcal F_{t_j}$, since there always exists $x_j^*(h)$ such that $R^{Best}(x^*_j(h))$ is arbitrarily close to the supremum, the proof is almost the same.}. By \Cref{proposition:best-response-trivial-lowerbound}, we know that $R^*\geq \frac{t_j}{t_j+t_{-j}}K$. Therefore in the main body of \Cref{alg:FPTAS}, there is at least one call to the DP procedure with an argument $K^*\in[\frac12R^*,R^*]$. We show that, in this case, the DP procedure finds a $\hat{x}_j(h)$ in $S$ such that $R^{Best}(\hat{x}_j(h))\geq (1-\epsilon')R^*$.

First, let $\tilde{x}_j(h)=\max(x^*_j(h),\lambda_x)$, then $\int_0^1\tilde{x}_j(h)dh\leq \int_0^1x^*_j(h)dh+\lambda_x\leq (1+\epsilon')t_j$. Next, construct $\bar{x}_j(h)=\min\{c_i:i\in\{1,2,\cdots,p\},c_i\geq \tilde{x}_j(h)\}$. One can see that $\forall h\in[0,1],\bar{x}_j(h)\leq (1+\epsilon')\tilde{x}_j(h)$, so $\int_0^1\bar{x}_j(h)dh\leq (1+\epsilon')\int_0^1\tilde{x}_j(h)dh\leq (1+\epsilon')^2t_j$. Since $\bar{x}_j(h)$ $1$-dominates $x^*_j(h)$, by \Cref{proposition:best-response-trivial-lowerbound}, $R^{Best}(\bar{x}_j(h))\geq R^*$. 

Let $\bar{B}_i=\bar{x}_j^{-1}(c_i)$, for $i=1,\cdots,p$, and let $\bar{B}_0=0$, we have $$\sum_{i=1}^p\int_{\bar{B}_{i-1}+x_{-j}^{-1}(c_i+0)}^{\bar{B}_{i}+x_{-j}^{-1}(c_i+0)}v(q)dq=R^{Best}(\bar{x}_j(h))\geq R^*$$. 
For $i=1,2,\cdots,p$, we reduce the length of the $i$-th segment to decrease the utility obtained by $[B_{i-1},B_i)$ to the closest multiple of $\lambda_v$. Formally, let $\hat{B}_0=0$, and for all $i=1,\cdots,p$, find $\hat{B}_{i}$ such that $$\int_{\hat{B}_{i-1}+x_{-j}^{-1}(c_i+0)}^{\hat{B}_{i}+x_{-j}^{-1}(c_i+0)}v(q)dq=\lambda_v\lfloor\frac1{\lambda_v}\int_{\bar{B}_{i-1}+x_{-j}^{-1}(c_i+0)}^{\bar{B}_{i}+x_{-j}^{-1}(c_i+0)}v(q)dq\rfloor.$$
It can be seen easily that $\hat{B}_i\leq \bar{B}_i$ holds for all $i=0,\cdots,p$.

Let $\hat{x}_j(h)$ be the interim allocation function corresponding with $\hat{B}_1,\cdots,\hat{B}_p$. Since $$\int_{\hat{B}_{i-1}+x_{-j}^{-1}(c_i+0)}^{\hat{B}_i+x_{-j}^{-1}(c_i+0)}v(q)dq>\int_{\bar{B}_{i-1}+x_{-j}^{-1}(c_i+0)}^{\bar{B}_{i}+x_{-j}^{-1}(c_i+0)}v(q)dq-\lambda_v,$$ we have $R^{Best}(\hat{x}_j(h))=\sum_{i=1}^p\int_{\hat{B}_{i-1}+x_{-j}^{-1}(c_i+0)}^{\hat{B}_{i}+x_{-j}^{-1}(c_i+0)}v(q)dq\geq R^*-p\lambda_v$. Recall that $\lambda_v=\frac{\epsilon'K^*}p$ and $K^*\leq R^*$, we have $R^{Best}(\hat{x}_j(h))\geq R^*-\epsilon'K^*\geq (1-\epsilon')R^*$.

Since $\hat{B}_i\leq\bar{B}_i$, we can see that $\int_0^1\hat{x}_j(h)dh\leq\int_0^1\bar{x}_j(h)dh\leq (1+\epsilon')^2t_j$. Considering the overestimation of the total budget requirement of $\hat{x}_j(h)$,   $\sum_{i=1}^p\lceil\frac{(\hat{B}_{i}-\hat{B}_{i-1})c_i}{\lambda_t}\rceil\lambda_t\leq\sum_{i=1}^p((\hat{B}_{i}-\hat{B}_{i-1})c_i+\lambda_t)\leq\int_0^1\hat{x}_j(h)dh+p\lambda_t\leq (1+\epsilon')^2t_j+\epsilon't_j\leq (1+\epsilon')^3t_j$.
Thus $\hat{x}_j(h)$ is in $S$, and consequently the utility of the best solution found by the DP procedure is at least $R^{Best}(\hat{x}_j(h))$.

By selecting the solution with highest utility among the solutions returned by all the calls to the DP procedure, we get a solution in $S$, still denoted by $\hat{x}_j(h)$, such that $R^{Best}(\hat{x}_j(h))\geq (1-\epsilon')R^*$. Let $\hat{B}_i$ denote $\hat{x}_j^{-1}(c_i)$, let $B_i=(1-\epsilon')^3\hat{B}_i$ be the final output of \Cref{alg:FPTAS}, corresponding with a $x_j(h)$, then $\int_0^1x_j(h)dh\leq (1-\epsilon')^3\int_0^1\hat{x}_j(h)dh\leq t_j$, so $x_j(h)\in\mathcal F_{t_j}$. On the other hand, as $x_j(h)$ $(1-\epsilon')^3$-dominates $\hat{x}_j(h)$, by \Cref{prop:relationship between w and b under C}, we have $R^{Best}(x_j(h))\geq (1-\epsilon')^3R^{Best}(\hat{x}_j(h))\geq(1-\epsilon')^4R^*> (1-4\epsilon')R^*=(1-\epsilon)R^*$.
\qed
\end{proof}
\section{Missing Proofs in Section \ref{section:safetylevel}}
\label{app:sec6}

\subsection{Proof of \Cref{proposition:safetylevel-trivial-upper-bound}}
\begin{proof}
If we regard all other designers as one designer, by \Cref{proposition:best-response-trivial-lowerbound}, for any $x_j(h)$, there exists a $x_{-j}(h)\in\mathcal F_{t_{-j}}$ such that $R^{Best}(x_{-j}(h)|x_j(h),v_j(q))\geq\frac{t_{-j}}{t_j+t_{-j}}\int_0^1v_{j}(q)dq$. It means that there exists a BNE of contestants, $(H_j(q),H_{-j}(q))$ $\in\mathcal H^*(x_j(h),x_{-j}(h))$, such that $\int_0^1v_j(q)dH_{-j}(q)\geq\frac{t_{-j}}{t_j+t_{-j}}\int_0^1v_j(q)dq$. Recall that $\int_0^1v_j(q)dH_{j}(q)+\int_0^1v_j(q)dH_{-j}(q)=\int_0^1v_j(q)dq$, so it follows that $\int_0^1v_j(q)dH_{j}(q)$ $\leq (1- \frac{t_{-j}}{t_j+t_{-j}})\int_0^1v_j(q)dq=\frac{t_j}{t_j+t_{-j}}\int_0^1v_j(q)dq$.
Therefore, the safety level of any $x_j(h)$ is at most $\frac{t_j}{t_j+t_{-j}}\int_0^1v_j(q)dq$, i.e., $\mathrm{SL}_{t_{-j},v_j(q)}(x_j(h))\leq \frac{t_j}{t_j+t_{-j}}\int_0^1v_j(q)dq$.
\qed
\end{proof}
\subsection{Proof of \Cref{lemma:example-simple-threshold-allocation}}

\begin{proof}
We consider an interim allocation function of the opponent: 
\begin{align*}
    x_{-j}(h)=\begin{cases}
    \frac{t_j}{r},&\text{if }h\leq \frac{t_{-j}}{t_j}r, \\
    0,&\text{if }h>\frac{t_{-j}}{t_j}r.
    \end{cases}
\end{align*}
Note that, at this time, we have 
\begin{align*}
    Q(x) = x_j^{-1}(x)+x_{-j}^{-1}(x) = \begin{cases}
    \frac{t_j+t_{-j}}{t_j}r,&\text{if }x\leq\frac{t_j}{r},\\
    0,&\text{if }x>\frac{t_j}{r}.
    \end{cases}
\end{align*}
Thus, we can obtain
\begin{align*}
    H_j^{Worst}(q)=\begin{cases}
        0,&\text{if }q<\frac{t_{-j}}{t_j}r;\\
        q-\frac{t_{-j}}{t_j}r,&\text{if }q\in[\frac{t_{-j}}{t_j}r,\frac{t_j+t_{-j}}{t_j}r];\\
        r,&\text{if }q>\frac{t_{-j}+t_j}{t_j}r.
    \end{cases}
\end{align*}
Therefore, the upper bound on the safety level of simple threshold allocation function can be shown as:
\begin{align*}
    &\mathrm{SL}_{t_{-j},v_j(q)}(x_j(h)) \leq \int_0^1v_j(q)dH_j^{Worst}(q)=\int_{\frac{t_{-j}}{t_j}r}^{\frac{t_j+t_{-j}}{t_j}r}v_j(q)dq\\
    \leq &rv_j(\frac{t_{-j}}{t_j}r)=\frac{t_j}{t_{-j}}\frac{t_{-j}r}{t_j}v_j(\frac{t_{-j}r}{t_j})\leq \frac{t_j}{t_{-j}}\max_{q\in[0,1]}qv_j(q)
\end{align*}
\qed
\end{proof}
\subsection{Extension of \Cref{example:simple-threshold-safetylevel}}
\label{example:extend-simple-threshold}
Let $\mathcal F_{t}^{Simple}$ denote the class of all simple threshold allocation functions under budget constraint $t$, i.e., $\mathcal F_{t}^{Simple}=\left\{x(h)=\begin{cases}\frac{t'}{r},&\text{if }h\in[0,r],\\0,&\text{if }h\in(r,1],\end{cases}:r\in(0,1],t'\in(0,t]\right\}$. Intuitively, every $x_j(h)\in F_{t}^{Simple}$ is a staircase function, which is a positive constant on the interval $[0,r]$. \Cref{example:simple-threshold-safetylevel} shows that, when a designer $j$'s strategy is restricted in $F_{t_j}^{Simple}$, her safety level can be arbitrarily small. We can extend this result to a broader class of interim allocation functions.

Fix any constant $D$, let $\mathcal F_{t}^{D-Simple}$ denote the class of interim allocation functions $x_j(h)$, such that there exists $0=A_0\leq A_1\leq A_2\leq \cdots\leq A_D\leq 1$ and $c_1\geq c_2\geq \cdots\geq c_D\geq0$, $\sum_{i=1}^D(A_{i}-A_{i-1})c_i\leq t$, and that $x_j(h)=\begin{cases}c_1,&\text{if }h\in[0,A_1],\\
c_i,&\text{if } h\in(A_{i-1},A_i],i\geq 2,\\
0, &\text{if } h\in(A_D,1]\end{cases}$. In other words, $\mathcal F_{t}^{D-Simple}$ consists of the staircase functions with (at most) $D$ positive stairs, under budget constraint $t$. Moreover, we define $\mathcal F_{t}^{C,D-dom}=\{x(h)\in\mathcal F_{t}:\exists \hat{x}(h)\in\mathcal F_{Ct}^{D-Simple},\forall h\in[0,1],x(h)\leq \hat{x}(h)\}$, be the class of all functions in $\mathcal F_{t}$ that is $1$-dominated by some function in $\mathcal F_{Ct}^{D-Simple}$, where $C>0$ is a constant.

Note that this is a much wider function class. For example, given any $D>0$ and $t>0$, for any $x(h)\in\mathcal F_{t}$ such that $x(0)\leq 2^{D-1}t$, we can construct $\hat{x}(h)$ by rounding up the value of $x(h)$ to the closest $2^it$, formally, $\hat{x}(h)=\min\{2^it:2^it\geq x(h),i=0,1,\cdots,D-1\}$. Note that $\hat{x}(h)\leq\max{ 2x(h),t}\leq 2x(h)+t$, so $\int_0^1\hat{x}(h)dh\leq 2\int_0^1x(h)dh+t\leq 3t$. Then it follows that $\hat{x}(h)\in\mathcal F_{3t}^{D-Simple}$, therefore $x(h)$ is in $\mathcal F_{t_j}^{3,D-dom}$.

We have the following proposition:
\begin{proposition}
Given $t_j,t_{-j}$, fix any constants $C$ and $D$, then the class $\mathcal F_{t_j}^{C,D-dom}$ does not guarantee a constant-competitive safety level, formally, for any $\epsilon>0$, there exists some value function $v_j(q)$, such that
$$\max_{x_j(h)\in \mathcal F_{t_j}^{C,D-dom}}\mathrm{SL}_{t_{-j},v_j(q)}(x_j(h))<\epsilon\frac{t_j}{t_j+t_{-j}}\int_0^1v(q)dq$$
\end{proposition}
\begin{proof}
Similar with \Cref{example:simple-threshold-safetylevel}, let $v_j(q)=\begin{cases}M,&\text{if }q\in[0,\frac1M],\\\frac1q,&\text{if }q\in(\frac1M,1]\end{cases}$, where $M>1$ is a large number to be determined. 

For any $x_j(h)\in \mathcal F_{t_j}^{C,D-dom}$, there is some $\hat{x}_j(h)\in F_{Ct_j}^{D-Simple}$ such that $\hat{x}_j(h)$ $1$-dominates $x_j(h)$.

Assume $\hat{x}_j(h)=\begin{cases}c_1,&\text{if }h\in[0,A_1],\\
c_i,&\text{if } h\in(A_{i-1},A_i],i\geq 2,\\
0, &\text{if } h\in(A_D,1]\end{cases}$, where $0=A_0\leq A_1\leq A_2\leq \cdots\leq A_D\leq 1$ and $c_1>c_2> \cdots>c_D> c_{D+1}=0$, $\sum_{i=1}^D(A_{i}-A_{i-1})c_i\leq Ct_j$.

We construct an interim allocation of the opponent by horizontal stretching method. For convenience, let $C'$ denote $\frac{t_{-j}}{Ct_j}$. Let $\hat{x}_{-j}(h)=\begin{cases}\hat{x}_j(h/C'),&\text{if }h\in[0,C'],\\0,&\text{if }h>C'\end{cases}$. Note that $\int_0^1\hat{x}_{-j}(h)dh\leq \frac1{C'}\int_0^1\hat{x}_j(h)dh\leq t_{-j}$, and $\hat{x}_{-j}(h)$ $C'$-dominates $\hat{x}_j(h)$.

Observe that $$\hat{x}_j^{-1}(x)=\begin{cases}0,&\text{if }x>c_1,\\
A_i,&\text{if }x\in(c_{i+1},c_i],i=1,2,\cdots,D,\\
1,&\text{if }x=0,
\end{cases}$$and $$\hat{x}_{-j}^{-1}(x)=\begin{cases}0,&\text{if }x>c_1,\\
C'A_i,&\text{if }x\in(c_{i+1},c_i],i=1,2,\cdots,D,\\
1,&\text{if }x=0.
\end{cases}$$
Let $\hat{H}_j^{Worst}(q)$ be the worst cumulative equilibrium behavior under $\hat{x}_j(q)$ and $\hat{x}_{-j}(h)$, as defined in \Cref{proposition:best-response-worst-and-best-equilibrium}, then with some calculations we have $$H_j^{Worst}(q)=\begin{cases}
A_{i-1},&\text{if }q\in[(1+C')A_{i-1},A_{i-1}+C'A_i],\\
q-C'A_i,&\text{if }q\in[A_{i-1}+C'A_i,(1+C')A_i],
\end{cases}$$
where $i=1,\cdots,D$.

For convenience assume $v_j(q)=0$ for any $q>1$, and we have \begin{align*}
&\int_0^1v_j(q)dH_{j}^{Worst}(q)\\
=&\sum_{i=1}^D\int_{A_{i-1}+C'A_i}^{(1+C')A_i}v_j(q)dq\\
\leq& \sum_{i=1}^D(A_{i}-A_{i-1})v_j(A_{i-1}+C'A_i)\\
\leq& \sum_{i=1}^DA_{i}v_j(A_{i-1}+C'A_i).
\end{align*}
Since $A_i\leq \frac1{C'}(A_{i-1}+C'A_i)$, we have $A_{i}v_j(A_{i-1}+C'A_i)\leq \frac1{C'}(A_{i-1}+C'A_i)v_j(A_{i-1}+C'A_i)\leq \frac1{C'}\max_{q\geq 0}qv_j(q)=\frac1{C'}$. Then we obtain that $\int_0^1v_j(q)dH_{j}^{Worst}(q)\leq \sum_{i=1}^DA_{i}v_j(A_{i-1}+C'A_i)\leq \frac{D}{C'}=\frac{CDt_j}{t_{-j}}$.

Since $\hat{x}_j(h)$ $1$-dominates $x_j(h)$, by \Cref{prop:relationship between w and b under C}, we have $$R^{Worst}(x_j(h)|\hat{x}_{-j}(h),v_j(q))\leq R^{Worst}(\hat{x}_j(h)|\hat{x}_{-j}(h),v_j(q))\leq\frac{CDt_j}{t_{-j}}.$$In addition, note that $\hat{x}_{-j}(h)\in\mathcal F_{t_{-j}}$, so by definition we have $$\mathrm{SL}_{t_{-j},v_j(q)}(x_j(h))\leq R^{Worst}(x_j(h)|\hat{x}_{-j}(h),v_j(q))\leq \frac{CDt_j}{t_{-j}}.$$
Therefore we obtain that $\max_{x_j(h)\in \mathcal F_{t_j}^{C,D-dom}}\mathrm{SL}_{t_{-j},v_j(q)}(x_j(h))\leq \frac{CDt_j}{t_{-j}}$.

Note that the discussion above is independent of the undetermined $M$ in the definition of $v_j(q)$. For any $\epsilon>0$, let $M=e^{\frac{CD(t_j+t_{-j})}{\epsilon t_{-j}}}$, then $\int_0^1v_j(q)dq=1+\ln M>\ln M=\frac{CD(t_j+t_{-j})}{\epsilon t_{-j}}$, and $\epsilon\frac{t_j}{t_j+t_{-j}}\int_0^1v_j(q)dq>\frac{CDt_j}{t_{-j}}\geq \max_{x_j(h)\in \mathcal F_{t_j}^{C,D-dom}}\mathrm{SL}_{t_{-j},v_j(q)}(x_j(h))$. This completes the proof.\qed
\end{proof}

\subsection{Proof of \Cref{theorem:16-competitive-safetylevel}}

\begin{proof}
Let $M = v_j(0)$ and $K = \int_0^1 v_j(q)dq$. The upper bound of $\frac{t_j}{t_j+t_{-j}}K$ on the safety level is given by \Cref{proposition:safetylevel-trivial-upper-bound}. We construct a strategy $x_j(h)$ with constant-competitive safety level as follows.

First, we define a sequence $\{q_i\}$, where $q_i=\sup\{q\in[0,1]:v_j(q)\geq \frac{M}{2^{i-1}}\}$, for $i=1,2,\cdots$, and define $q_0=0$. Clearly, the sequence satisfies $q_0\leq q_1\leq q_2\leq\cdots$ and $\lim_{i\to+\infty}q_i\leq1$.
Then, we let $s_i=\int_{q_{i-1}}^{q_i}v_j(q)dq$, $c_i=\frac{CM(t_j+t_{-j})}{K2^{i-1}}$, and $a_i=\frac{t_js_i}{Kc_i}=\frac{t_js_i}{(t_j+t_{-j})C}/\frac{M}{2^{i-1}}$, where $C>0$ is an undetermined constant and will be determined later.

Let $A_i=\sum_{i'=1}^ia_{i'}$, we can construct $x_j(h)$ as:
\begin{align}\label{eqn: constrcution of safety level}
    x_j(h)=\begin{cases}
    c_i,&\text{if }\exists i\in\{1,2,\cdots\}\text{ such that }A_{i-1}\leq h<A_i; \\
    0,& \text{otherwise}.
\end{cases}
\end{align} 
It can be checked that $\int_0^1x_j(h)dh=\sum_{i=1}^{+\infty}a_ic_i=\sum_{i=1}^{+\infty}\frac{t_js_i}{K}=t_j$, so $x_j(h)\in\mathcal F_{t_j}$.

Since $s_i=\int_{q_{i-1}}^{q_i}v_j(q)dq \leq \frac{M}{2^{i-2}}(q_{i}-q_{i-1})$, we have \begin{align*}
    a_i=\frac{t_js_i}{(t_j+t_{-j})C}/\frac{M}{2^{i-1}}\leq \frac{2t_j}{C(t_j+t_{-j})}(q_i-q_{i-1}).
\end{align*}
When $C\geq 2$, $A_i=\sum_{i'=1}^ia_{i'}\leq \frac{2t_j}{C(t_j+t_{-j})} q_i\leq q_i$ holds.

Then, we prove the lower bound on the utility gained by $x_j(h)$, against any strategy of the opponent. For any $x_{-j}(h)\in\mathcal F_{t_{-j}}$, let $B_i$ denote $x_{-j}^{-1}(c_i)$, and define $B_0=0$. If $B_i>0$, we have $x_{-j}(h)\geq c_i$ holds for any $h\in[0,B_i)$. Define  $b_i=B_i-B_{i-1}$, then we have $$\sum_{i=1}^{+\infty}c_i b_i \leq\sum_{i=1}^{+\infty}\int_{B_{i-1}}^{B_i}x_{-j}(h)dh= \int_0^1 x_{-j}(h)dh \leq t_{-j}.$$

We say $i\in\{1,2,\cdots\}$ is \emph{good}, if $A_i+B_i\leq q_i$. We define the set of \emph{good} $i$'s as $\mathrm{GD}=\{i, A_i+B_i\leq q_i\}$. For any GCEB $\vec{H}(q)=(H_j(q),H_{-j}(q))\in\mathcal H^*(x_j(h),x_{-j}(h))$ and any $i\in\mathrm{GD}$, we have $Q(c_i)=x^{-1}_j(c_i)+{x}^{-1}_{-j}(c_i)=A_i+B_i\leq q_i$ and $H_j(Q(c_i))=x^{-1}(c_i)=A_i$.

Let $q'_i=Q(c_i)=A_i+B_i$, for $i=1,2, \cdots$ and $q'_0=0$. For a good $i$, we have $q'_i\leq q_i$, so $v(q)\geq \frac{M}{2^{i-1}}$ for any $q\in[0,q'_i)$. Thus, we have
\begin{align*}
\int_{q'_{i-1}}^{q'_i}v(q)dH_j(q)\geq \frac{M}{2^{i-1}}\int_{q'_{i-1}}^{q'_i}dH_j(q)
=\frac{M}{2^{i-1}}(A_i-A_{i-1})
=\frac{M}{2^{i-1}}a_i.
\end{align*}
Recall that $a_i=\frac{t_js_i}{(t_j+t_{-j})C}/\frac{M}{2^{i-1}}$, so $\int_{q'_{i-1}}^{q'_i}v(q)dH_j(q)\geq\frac{M}{2^{i-1}}a_i\geq\frac{t_js_i}{(t_j+t_{-j})C}$.
It follows that 
\begin{equation}\label{ieq: welfare of j}
\int_0^1v(q)dH_j(q)
\geq \sum_{i\in\mathrm{GD}}\int_{q'_{i-1}}^{q'_i}v(q)dH_j(q)
\geq \frac{t_j}{(t_j+t_{-j})C}\sum_{i\in\mathrm{GD}}s_i.
\end{equation}

Note that $K=\sum_{i=1}^{+\infty}s_i=\sum_{i\in\mathrm{GD}}s_i+\sum_{i\notin\mathrm{GD}}s_i$. Now we show that $\sum_{i\notin\mathrm{GD}}s_i$ is actually restricted due to the opponent's budget constraint.

For any $i\notin\mathrm{GD}$, we have $A_i+B_i>q_i$, which implies that $\sum_{i'=1}^{i}b_{i'}=B_i>q_i-A_i\geq (1-\frac{2t_j}{(t_j+t_{-j})C})q_i$.
Since $c_{i+1}=\frac12c_{i}$, we have
$$
    \sum_{i=1}^{+\infty}B_ic_i=\sum_{i=1}^{+\infty}\sum_{i'=1}^ib_{i'}c_i=\sum_{i=1}^{+\infty}b_i\sum_{i'=i}^{+\infty}c_{i'}
    =\sum_{i=1}^{+\infty}b_ic_i\sum_{i'=i}^{+\infty}2^{-(i'-i)}=2\sum_{i=1}^{+\infty}b_ic_i\leq 2t_{-j}.
$$
On the other hand,
\begin{align*}
\sum_{i=1}^{+\infty}B_ic_i \geq&\sum_{i\notin\mathrm{GD}}B_ic_i\\
>&\sum_{i\notin\mathrm{GD}}c_i\left(1-\frac{2t_j}{(t_j+t_{-j})C}\right)q_i\\
=&\left(1-\frac{2t_j}{(t_j+t_{-j})C}\right)\sum_{i\notin\mathrm{GD}}\frac{C(t_j+t_{-j})}{K}\frac{M}{2^{i-1}}q_i\\
=& \frac{C(t_j+t_{-j})-2t_j}{K}\sum_{i\notin\mathrm{GD}}\frac{M}{2^{i-1}}q_i.
\end{align*}
Combining the above two inequalities, we get $\sum_{i\notin\mathrm{GD}}\frac{M}{2^{i-1}}q_i\leq\frac{2t_{-j}K}{C(t_j+t_{-j})-2t_j}$. Recall that $s_i\leq \frac{M}{2^{i-2}(q_i-q_{i-1})}$ and we have
\begin{align*}
    \sum_{i\notin\mathrm{GD}}s_i\leq \sum_{i\notin\mathrm{GD}}\frac{M}{2^{i-2}}(q_i-q_{i-1}) 
    \leq \sum_{i\notin\mathrm{GD}}\frac{M}{2^{i-2}}q_i\leq \frac{4t_{-j}K}{C(t_j+t_{-j})-2t_j}.
\end{align*}
Thus, we obtain
\begin{align}\label{ieq: si of good}
    \sum_{i\in\mathrm{GD}}s_i\geq K-\sum_{i\notin\mathrm{GD}}s_i\geq \left(1-\frac{4t_{-j}}{C(t_j+t_{-j})-2t_j}\right)K.
\end{align}
By the inequalities (\ref{ieq: welfare of j}) and (\ref{ieq: si of good}), we have 
\begin{align*}
\int_0^1v(q)dH_j(q)
\geq& \frac{t_j}{(t_j+t_{-j})C}\sum_{i\in\mathrm{GD}}s_i\\
\geq& \left(1-\frac{4t_{-j}}{C(t_j+t_{-j})-2t_j}\right)\frac{t_j}{(t_j+t_{-j})C}K.
\end{align*}
Taking $C=\frac{\sqrt{4t_{-j}(2t_j+4t_{-j})}+2t_j+4t_{-j}}{t_j+t_{-j}}$, the above inequality changes to 
\begin{align}\label{ieq:safefy level t}
    \int_0^1v(q)dH_j(q)\geq\frac{t_j}{2t_j+8t_{-j}+2\sqrt{4t_{-j}(2t_j+4t_{-j})}}K.
\end{align}
If we set $p=\frac{t_j}{t_j+t_{-j}}$, the inequality (\ref{ieq:safefy level t}) is rewritten as 
\begin{align*}
    \int_0^1v(q)dH_j(q)\geq\frac{1}{8-6p+4\sqrt{(1-p)(4-2p)}}\frac{t_j}{t_j+t_{-j}}K. 
\end{align*}
Denote $sl(p)= 8-6p+4\sqrt{(1-p)(4-2p)}$. By now, we can see that $x_j(h)$ constructed by (\ref{eqn: constrcution of safety level}) is $sl(\frac{t_j}{t_j+t_{-j}})$-competitive, given the budgets $t_j$ and $t_{-j}$. Moreover, since $sl(p)$ is decreasing on $[0,1]$, it reach the maximum $sl(0)=16$ at $p=0$. In summary, 
$x_j(h)$ is a $16$-competitive strategy.
\qed
\end{proof}
\section{Missing Proofs in Section \ref{section:implement}}
\label{app:sec7}

Before proving \Cref{theorem:approximating-interim-allocation-with-rank-by-skill}, we first introduce two technical lemmas.
\subsection{Two Technical Lemmas}
Let $\xi_{n,k}(h)=\sum_{l=1}^k\binom{n-1}{l-1}h^{l-1}(1-h)^{n-l}$.
We have the following lemma:
\begin{lemma}\label{lemma:simple-contest-chernoff-bound}
For any $\delta>0$, for any $n,k$ such that $k>6\ln\frac1{\delta}$, let $\beta=\sqrt{\frac{6\ln \frac1{\delta}}{k}}$. Then for any $h\in [0,\frac{k}{(1+\beta)(n-1)}]$, $\xi_{n,k}(h)>1-\delta$.
\end{lemma}
\begin{proof}
Since $\xi_{n,k}(h)$ is non-increasing in $h$, without loss of generality we assume $h=\frac{k}{(1+\beta)(n-1)}$. When $k>6\ln\frac1{\delta}$, we have $\beta<1$.

Suppose a random variable $B$ which follows the binomial distribution $B(n-1,h)$.
Note that $\xi_{n,k}(h)=\Pr[B<k]=1-\Pr[B\geq k]$.

Since $\E[B]=(n-1)h=\frac{k}{1+\beta}$, by Chernoff bound (\cite{AV79}), we have
$$\Pr[B\geq k]=\Pr[B\geq(1+\beta)\E[B]]\leq e^{-\frac{\beta^2\E[B]}{\beta+2}}=e^{-\frac{k\beta^2}{(\beta+2)(\beta+1)}}.$$
Because $\beta<1$, we have $\frac{(\beta+2)(\beta+1)}{\beta^2}=1+3/\beta+2/\beta^2
< 6/\beta^2\leq\frac{\ln\frac1{\delta}}{k}$, therefore $\Pr[B\geq k]\leq e^{-\frac{k\beta^2}{(\beta+2)(\beta+1)}}<e^{-\ln\frac1{\delta}}=\delta$, and $\xi_{n,k}(h)=1-\Pr[B\geq k]\geq 1-\delta$.
\qed
\end{proof}




\begin{lemma}\label{lemma:delta-C-dominate-bestresponse}
For $C\in(0,1]$, we say $x(h)$ $\delta,C$-dominates $\hat{x}(h)$, if
$$x(Ch)\geq \hat{x}(h),\forall h\in[\delta,1].$$
If $x_j(h)$ $\delta,C$-dominates $\hat{x}(h)$, i.e., $\forall h\in[\delta,1],x(Ch)\geq \hat{x}(h)$, for some $\delta>0,C\in(0,1]$, then for any $x_{-j}(h)$, for any $v_j(q)$, $R^{Best}(x_j(h)|x_{-j}(h),v_j(q))\geq CR^{Best}(\hat{x}_j(h)|x_{-j}(h),v_j(q))-\int_0^{C\delta} v_j(q)dq$.
\end{lemma}
\begin{proof}
Let $M=\max(x_{-j}(0),\hat{x}_j(0))+1$.
Construct $\bar{x}(h)=\begin{cases}M,&\text{if }h\in[0,\delta];\\
x_j(h),&\text{if }h\in(\delta,1]\end{cases}$.

Since $\bar{x}(h)$ $C$-dominates $\hat{x}(h)$, by \Cref{prop:relationship between w and b under C}, $R^{Best}(\bar{x}_j(h))\geq CR^{Best}(\hat{x}_j(h))$.

Let $H_j^{Best}(q)$ be the best cumulative equilibrium for designer $j$ in $\mathcal H_j^*(x_j(h),x_{-j}(h))$, and $\hat{H}_j^{Best}(q)$ be the best cumulative equilibrium in $\mathcal H_j^*(\bar{x}_j(h),x_{-j}(h))$, as defined in \Cref{proposition:best-response-worst-and-best-equilibrium}.

Let $Q(x)=x_j^{-1}(x)+x_{-j}^{-1}(x)$, $\bar{Q}(x)=\bar{x}_j^{-1}(x)+x_{-j}^{-1}(x)$. Since $\bar{x}_j(h)\geq x_j(h)$, we have $\bar{x}_j^{-1}(x)\geq x_j^{-1}(x)$, and $\bar{Q}(x)\geq Q(x)$.

Let $X^*=x_j(C\delta)$. For any $X\leq X^*$, $\bar{x}_j^{-1}(X)=x_j^{-1}(X)\geq C\delta$.
For any $q\in[Q(X^*+0),Q(X^*)]$, $H_j^{Best}(q)=\min(q-x_{-j}^{-1}(X^*+0),x_j^{-1}(X^*))$. For any $q\in[\bar{Q}(X^*+0),\bar{Q}(X^*)]$, $\bar{H}_j^{Best}(q)=\min(q-x_{-j}^{-1}(X^*+0),\bar{x}_j^{-1}(X^*))$.

Let $q^*=x_{-j}^{-1}(X^*+0)+C\delta$, we have $q^*-x_{-j}^{-1}(X^*+0)=C\delta\leq \bar{x}_j^{-1}(X^*)=x_j^{-1}(X^*)$. Since $X^*=x_j(C\delta)=\bar{x}_j(C\delta)$, we have $x_j^{-1}(X^*+0)\leq C\delta$ and $\bar{x}_j^{-1}(X^*+0)=C\delta$. Then $Q(X^*+0)=x_j^{-1}(X^*+0)+x_{-j}^{-1}(X^*+0)\leq C\delta+x_{-j}^{-1}(X^*+0)=q^*$, and $\bar{Q}(X^*+0)=\bar{x}_j^{-1}(X^*+0)+x_{-j}^{-1}(X^*+0)=C\delta+x_{-j}^{-1}(X^*+0)=q^*$. Therefore $H_j^{Best}(q^*)=\min(q^*-x_{-j}^{-1}(X^*+0),x_j^{-1}(X^*))=C\delta$, and $\bar{H}_j^{Best}(q^*)=\min(q^*-x_{-j}^{-1}(X^*+0),\bar{x}_j^{-1}(X^*))=C\delta$.

For any $X'\leq X^*$, since $\bar{x}_j^{-1}(X')=x_j^{-1}(X')$, we have $Q(X')=\bar{Q}(X')$, thus for any $q'\geq q$, $Q^{-1}(q')=\bar{Q}^{-1}(q')$, so $H_j^{Best}(q')=\bar{H}_j^{Best}(q')$. So we obtain
\begin{align*}\int_0^1v_j(q)dH_j^{Best}(q)
\geq& \int_{q^*}^1v_j(q)dH_j^{Best}(q)\\
=&\int_{q^*}^1v_j(q)d\bar{H}_j^{Best}(q)\\
\geq&\int_{0}^1v_j(q)d\bar{H}_j^{Best}(q)-\int_{0}^{q^*}v_j(q)d\bar{H}_j^{Best}(q)
\end{align*}
Let $g(q)=\min\{q,C\delta\}$, for any $q\in[0,q^*]$, since $\bar{H}_j^{Best}(q)\leq q$ and $\bar{H}_j^{Best}(q)\leq \bar{H}_j^{Best}(q^*)=C\delta$, we have $g(q)\geq \bar{H}_j^{Best}(q)$. By \Cref{lemma:integral-compare},  $\int_0^{q^*}v_j(q)d\bar{H}_j^{Best}(q)\leq \int_0^{q^*}v_j(q)dg(q)=\int_0^{C\delta}v_j(q)dq$. Recall that $\int_{0}^1v_j(q)d\bar{H}_j^{Best}(q)\geq CR^{Best}(\hat{x}_j(h))$, we obtain $\int_0^1v_j(q)dH_j^{Best}(q)\geq CR^{Best}(\hat{x}_j(h))-\int_0^{C\delta}v_j(q)dq$.
\qed
\end{proof}

\subsection{Proof of \Cref{theorem:approximating-interim-allocation-with-rank-by-skill}}
\begin{proof}

First we note that for any $\vec{w}=(w_1,\cdots,w_n)$, for arbitrarily small $\epsilon>0$, we can construct $\vec{w}'=((1-\epsilon)w_1+\epsilon t_j,\cdots,(1-\epsilon)w_n+\epsilon t_j)$, then $x_{\vec{w}'}(h)$ $(1-\epsilon)$-strongly-dominates $x_{\vec{w}}(h)$, and then $R^{Worst}(x_{\vec{w}'}(h))\geq (1-\epsilon) R^{Best}(x_{\vec{w}}(h))$. Therefore we only need to construct $\vec{w}$ such that
$R^{Best}(x_{\vec{w}}(h)|x_{-j},v_j(q))\geq (1-r(n)) R^{Best}(x_j(h)|x_{-j},v_j(q))$

Let $C\in(0,1)$ be an undetermined parameter.
For sufficiently large $n$, let $\delta=\frac1n$, $k_1=\lfloor n^{\frac23}\rfloor$, $\beta=\sqrt{\frac{6\ln \frac1{\delta}}{k_1}}$. 

For any $i\geq k_1$, take $p(i)=\min(\lceil(1+\beta)i\rceil,n)$.

When $p(i)<n$, $\frac{i}{n-1}\leq \frac{p(i)}{(1+\beta)(n-1)}\leq \frac{p(i)}{(1+\sqrt{\frac{6\ln \frac1{\delta}}{p(i)}})(n-1)}$, so by \Cref{lemma:simple-contest-chernoff-bound}, we have $\xi_{n,p(i)}(\frac{i}{n-1})\geq 1-\delta$, and $\int_0^1\xi_{n,p(i)}(h)dh=\frac{p(i)}{n}\leq \frac1n+\frac{i(1+\beta)}{n}\leq \frac{(i+1)(1+\beta)}{n}$.

When $p(i)=n$, $\xi_{n,p(i)}(\frac{i}{n-1})=1$, and $i>\frac{n-1}{(1+\beta)}$, so $\int_0^1\xi_{n,p(i)}(h)dh=1\leq (1+\beta)\frac{i}{n-1}$.

In either case, $\xi_{n,p(i)}(\frac{i}{n-1})\geq 1-\delta$ and $\int_0^1\xi_{n,p(i)}(h)dh\leq (1+\beta)\frac{i+1}{n}$.

Given $x_j(h)$, let $\bar{x}(h)=\begin{cases}
x_j(h/C),&\text{if }h\leq C,\\
0,&\text{if }h>C,
\end{cases}$, then $\bar{x}(h)$ $C$-dominates $x_j(h)$. Note that $\bar{x}(1)=0$.

For $k=1,2,\cdots,n$, let
$$d_k=\sum_{i\in\{k_1,k_1+1,\cdots,n-1\}:p(i)=k}\frac{1}{1-\delta}\left(\bar{x}(\frac{i-1}{n-1})-\bar{x}(\frac{i}{n-1})\right),$$
let $w_k=\sum_{j=k}^{n}d_k$, $\vec{w}=(w_1,w_2,\cdots,w_n)$.
Then $x_{\vec{w}}(h)=\sum_{k=1}^nw_k\binom{n-1}{k-1}h^{k-1}(1-h)^{n-k}=\sum_{k=1}^nd_k\xi_{n,k}(h)$.

By previous discussion we know that for any $i\in\{k_1,k_1+1,\cdots,n-1\}$, $\forall h\in[\frac{i-1}{n-1},\frac{i}{n-1}]$, $x_{\vec{w}}(h)\geq \sum_{k=p(i)}^{n}d_k\xi_{n,k}(h)\geq \sum_{i'=i}^{n-1}(\bar{x}(\frac{i'-1}{n-1})-\bar{x}(\frac{i'}{n-1}))=\bar{x}(\frac{i-1}{n-1})\geq\bar{x}(h)$, so $\forall h\in[\frac{k_1-1}{n-1},1]$, $x_{\vec{w}}(h)\geq \bar{x}(h)$. In other words, $\forall h\in[\frac{k_1-1}{C(n-1)},1]$, $x_{\vec{w}}(Ch)\geq \bar{x}(Ch)\geq x_j(h)$, so $x_{\vec{w}}(h)$ $\frac{k_1-1}{C(n-1)},C$-dominates $x(h)$.

Next we calculate the budget requirement of $x_{\vec{w}}(h)$.
\begin{align*}
&\int_0^1x_{\vec{w}}(h)dh\\
=&\sum_{k=1}^nd_k\int_0^1\xi_{n,k}(h)dh\\
=&\frac{1}{1-\delta}\sum_{i=k_1}^{n-1}\left(\bar{x}(\frac{i-1}{n-1})-\bar{x}(\frac{i}{n-1})\right)\int_0^1\xi_{n,p(i)}(h)dh\\
\leq &\frac{1}{1-\delta}\sum_{i=k_1}^{n-1}\left(\bar{x}(\frac{i-1}{n-1})-\bar{x}(\frac{i}{n-1})\right)(1+\beta)\frac{i+1}{n}\\
=& \frac{1+\beta}{1-\delta}\sum_{i=k_1}^{n-1}\left(\bar{x}(\frac{i-1}{n-1})-\bar{x}(\frac{i}{n-1})\right)\frac{i+1}{n}\\
\leq&\frac{1+\beta}{1-\delta}\frac{k_1+1}{k_1-1}\sum_{i=k_1}^{n-1}\left(\bar{x}(\frac{i-1}{n-1})-\bar{x}(\frac{i}{n-1})\right)\frac{i-1}{n}\\
=&\frac{1+\beta}{1-\delta}\frac{k_1+1}{k_1-1}
\left(\bar{x}(\frac{k_1-1}{n-1})\frac{k_1-1}n+\sum_{i=k_1}^{n-2}\bar{x}(\frac{i}{n-1})\frac{1}{n}\right)\\
\leq&\frac{1+\beta}{1-\delta}\frac{k_1+1}{k_1-1}\frac{n-1}{n}
\int_0^1\bar{x}(h)dh\\
=&\frac{1+\beta}{1-\delta}\frac{k_1+1}{k_1-1}\frac{n-1}{n}C\int_0^1x_j(h)dh
\end{align*}

Take $C=\frac{1-\delta}{1+\beta}\frac{k_1-1}{k_1+1}\frac{n}{n-1}$, then $\int_0^1x_{\vec{w}}(h)dh\leq \int_0^1x_j(h)dh$.

Since $C=\frac{1-\frac1n}{1+\sqrt{\frac{6\ln n}{k_1}}}\frac{k_1-1}{k_1+1}\frac{n}{n-1}=\frac1{1+\sqrt{\frac{6\ln n}{k_1}}}\frac{k_1-1}{k_1+1}$, and $x_{\vec{w}}(h)$ $\frac{k_1-1}{C(n-1)},C$-dominates $x(h)$, by \Cref{lemma:delta-C-dominate-bestresponse}, $R^{Best}(x_{\vec{w}}(h);x_{-j}(h);v_j(q))\geq CR^{Best}(x_j(h);x_{-j}(h);v_j(q))-\int_0^{C\frac{k_1-1}{C(n-1)}} v_j(q)dq$.
By assumption $R^{Best}(x_j(h);x_{-j}(h);v_j(q))\geq \frac1{D}K^*$, so we have $\int_0^{C\frac{k_1-1}{C(n-1)}} v_j(q)dq=\int_0^{\frac{k_1-1}{n-1}} v_j(q)dq\leq M\frac{k_1-1}{n-1}\leq \frac{DM}{K^*}\frac{k_1-1}{n-1}R^{Best}(x_j(h);x_{-j}(h);v_j(q))$. Thus we have $\frac{R^{Best}(x_{\vec{w}}(h);x_{-j}(h);v_j(q))}{R^{Best}(x_j(h);x_{-j}(h);v_j(q))}\geq C-\frac{DM}{K^*}\frac{k_1-1}{n-1}
=\frac1{1+\sqrt{\frac{6\ln n}{k_1}}}$ $\frac{k_1-1}{k_1+1}-\frac{DM}{K^*}\frac{k_1-1}{n-1}
=(1-O(\sqrt{\frac{\ln n}{n^{\frac23}}}))(1-O(\frac1{n^{\frac23}}))-\frac{DM}{K^*}O(n^{-\frac13})=1-O((\frac{DM}{K^*}+\sqrt{\ln n})n^{-\frac13})$
\qed
\end{proof}

\end{document}